\documentclass{new_tlp}
\usepackage{graphicx}
\usepackage{amsmath,amsfonts}
\usepackage[obeyFinal]{todonotes}
\usepackage{booktabs}
\usepackage{xspace}
\usepackage{mathtools}

\newtheorem{theorem}{Theorem}
\newtheorem{definition}[theorem]{Definition}
\newtheorem{corollary}[theorem]{Corollary}

\newtheorem{example}[theorem]{Example}

\begin{document}
\newcommand{\elDL}{\mathcal{EL}}
\newcommand{\belDL}{\mathcal{BE\!L}}
\newcommand{\ALC}{\ensuremath{\mathcal{ALC}}\xspace}
\newcommand{\EL}{\ensuremath{\mathcal{E\!L}}\xspace}
\newcommand{\BALC}{\ensuremath{\mathcal{BALC}}\xspace}
\newcommand{\BEL}{\ensuremath{\mathcal{BEL}}\xspace}

\newcommand{\ExpTime}{{\sc ExpTime}\xspace}

\newcommand{\Amf}{\ensuremath{\mathfrak{A}}\xspace}

\newcommand{\Amc}{\ensuremath{\mathcal{A}}\xspace}
\newcommand{\Bmc}{\ensuremath{\mathcal{B}}\xspace}
\newcommand{\Imc}{\ensuremath{\mathcal{I}}\xspace}
\newcommand{\Jmc}{\ensuremath{\mathcal{J}}\xspace}
\newcommand{\Kmc}{\ensuremath{\mathcal{K}}\xspace}
\newcommand{\Omc}{\ensuremath{\mathcal{O}}\xspace}
\newcommand{\Pmc}{\ensuremath{\mathcal{P}}\xspace}
\newcommand{\Tmc}{\ensuremath{\mathcal{T}}\xspace}
\newcommand{\Vmc}{\ensuremath{\mathcal{V}}\xspace}

\newcommand{\val}{\ensuremath{\mathsf{val}}\xspace}
\newcommand{\pa}{\ensuremath{\mathsf{pa}}\xspace}

\newcommand{\tboxDL}{\mathcal{T}}
\newcommand{\aboxDL}{\mathcal{A}}
\newcommand{\kbDL}{\mathcal{K}}
\newcommand{\bnDL}{\mathcal{B}}

\newcommand{\domainDL}{\Delta^{\mathcal{I}}}
\newcommand{\interpFuncDL}{\cdot^{\mathcal{I}}}

\newcommand{\vDomainDL}{\Delta^{\mathcal{V}}}
\newcommand{\vInterpFuncDL}{\cdot^{\mathcal{V}}}
\newcommand{\valFuncDL}{v^{\mathcal{V}}}

\newcommand{\modelDL}{\mathcal{I}}
\newcommand{\vModelDL}{\mathcal{V}}
\newcommand{\probModelDL}{\mathcal{P}}

\newcommand{\balcStmnt}[2]{#1 ^ #2}

\newcommand{\et}[1]{\ensuremath{\mathsf{#1}}}

\newcommand{\red}[1]{\textcolor{red}{#1}}
\newcommand{\modify}[1]{{#1}}

\title{The Probabilistic Description Logic \BALC}
%
%
\author[Leonard Botha, Thomas Meyer and Rafael Pe{\~n}aloza]
	{LEONARD BOTHA, THOMAS MEYER \\
	University of Cape Town and CAIR, South Africa
	 \and RAFAEL PE\~NALOZA\thanks{Part of this work was carried out while this author was at
the Free University of Bozen-Bolzano, Italy.} \\
	University of Milano-Bicocca, Italy}
%
%
%
\maketitle             
\begin{abstract} 
Description logics (DLs) are well-known knowledge representation formalisms focused on the representation
of terminological knowledge. Due to their first-order semantics, these languages (in their classical form) are 
not suitable for representing and handling uncertainty.
A probabilistic extension of a light-weight DL was recently proposed for dealing
with certain knowledge occurring in uncertain contexts. In this paper, we continue that line of research by 
introducing the Bayesian extension \BALC of the propositionally closed DL \ALC. We present a tableau-based 
procedure for deciding consistency, and adapt it to solve other probabilistic, contextual, and general inferences 
in this logic. We also show that all these problems remain \ExpTime-complete, the same as reasoning in the 
underlying classical \ALC. Under consideration in Theory and Practice of Logic Programming (TPLP).
\end{abstract}
\section{Introduction}

Description logics (DLs)~\cite{BCea-07} are a family of logic-based knowledge representation formalisms 
designed to describe the terminological knowledge of an application domain. Description Logics have been 
successfully applied to model several domains, with some particularly successful applications coming 
from the biomedical sciences. This is due to their clear syntax,
formal semantics, the existence of efficient reasoners, and their expressivity. However, in their classical form,
these logics are not capable of dealing with uncertainty, which is an unavoidable staple in real-world knowledge. 
To overcome this limitation, several probabilistic extensions of DLs have been suggested in the literature. 
The landscape of probabilistic extensions of DLs is too large to be covered in detail in this work. These logics
differentiate themselves according to their underlying logical formalism, their interpretation of probabilities,
and the kind of uncertainty that they are able to express. For a relevant 
survey,
where all these differences are showcased, we refer the interested reader to \cite{LuSt-JWS08}.
\modify{More recent work where probabilistic DLs are discussed can be found in \cite{CeLu-18,GJLR-17}.}

A \modify{related} probabilistic DL, called \BEL \modify{\cite{CePe-JAR16}}, is the Bayesian extension of the light-weight 
\EL \cite{BaBL05}. 
This logic focuses on modelling
certain knowledge that holds only in some contexts, together with uncertainty about the current context. 
\modify{\BEL is based on a \emph{subjective} interpretation of probabilities or, in Halpern's terminology, it corresponds
to a Type II probabilistic logic \cite{Halp-AIJ90}.}
One advantage of the formalism underlying \BEL is that it separates the contextual knowledge, which
is \emph{de facto} a classical ontology, from the likelihood of observing this context, which can be influenced
by external factors. We present a simple example
of the importance of contextual knowledge. Consider the knowledge of construction techniques and materials
that vary through time. In the context of a modern house asbestos and lead pipes are not observable, 
while in some classes of houses, built during the 1970s, we observe both. However, in all contexts we know
that asbestos and lead in drinking water have grave health effects. When confronted with a random
house, one might not know to which of these contexts it belongs, and by extension whether it is safe to
live in, or drink the water that flows through its pipes. Still, construction data may be used to derive the 
probabilities of these contexts.

To allow
for complex probabilistic relationships between the contexts without needing to result to 
\modify{incompatible} independence
assumptions, their joint probability distribution is encoded via 
a Bayesian network (BN)~\cite{Pear-85}. 
This logic is closely related to the probabilistic extension of DL-Lite \cite{ACKZ09} proposed previously in~\cite{AmFL08}, but uses 
a less restrictive semantics which resembles more the open-world assumption from DLs (for a discussion on the 
differences between the semantics of these logics, see~\cite{CePe-JAR16}).
Another similar proposal is Probabilistic Datalog$^\pm$~\cite{GottlobLMS13}, with the difference that uncertainty is 
represented via a Markov Logic Network, instead of a BN.
Since the introduction of \BEL, the main notions behind it have been 
generalised to arbitrary ontology languages~\cite{Ceyl-18}. However, it has also been shown that efficient and 
complexity-optimal reasoning methods can only be achieved by studying the properties of each underlying
ontology language~\cite{CePe-JAR16}.
\modify{Finally, another family of probabilistic DLs using Type II semantics was proposed in \cite{GJLR-17}, which we refer to as Prob-DLs. The biggest difference between Bayesian DLs and Prob-DLs is that the latter models uncertain concepts, while the 
former models uncertain knowledge. For example, Prob-DLs can refer to the class of all individuals having a high probability of a disease infection, but cannot express that an entailment holds with a given probability.}

In this paper, we continue with that line of research and study the
Bayesian extension of the propositionally closed DL \ALC. As our main result, we present an algorithm, based
on a glass-box modification of the classical tableaux method for reasoning in \ALC. 
\modify{Our algorithm is able to derive those contexts which are needed to determine the (in)consistency of a \BALC knowledge base.}
Using this algorithm, we then describe an effective method
for deciding consistency of a \BALC knowledge base. We also provide a tight \ExpTime complexity bound for 
this problem.

This is followed by a study of several crisp and probabilistic
variants of the standard DL decision problems; namely, concept satisfiability, subsumption, and instance checking.
Interestingly, our work shows that all our problems can be reduced to some basic computations over a context
describing inconsistency, and hence are \ExpTime-complete as well. These complexity bounds are not 
completely surprising, given the high complexity of the classical \ALC. However, our tableaux-based algorithm
has the potential to behave better in practical scenarios.
This work details and deepens results that have previously been presented in~\cite{Both-18,BoMP18,BoMP19}

\section{Preliminaries}

We start by providing a brief introduction to Bayesian networks and the description logic (DL) \ALC, which form the basis for 
the probabilistic DL \BALC. For a more detailed presentation of these topics, we refer the interested reader to
\cite{Darw09,BHLS17}

\subsection{Bayesian Networks} 
Bayesian networks (BNs) are graphical models, which are used for representing the joint probability distribution 
(JPD) of several discrete random variables in a compact manner~\cite{Pear-85}. Before introducing these models
formally, we need a few definitions.

Given a random variable $X$, let $\val(X)$ denote the set of values that $X$ can take. For the scope of this paper,
we consider only random variables $X$ such that $\val(X)$ is finite. For
$x\in val(X)$, $X=x$ denotes the \emph{valuation} of $X$ taking the value $x$. This notation is 
extended to sets of variables in the obvious way.
Given a set of random variables $V$\!, a \emph{world} $\omega$ is a set of valuations containing exactly one 
valuation for every random variable $X \in V$\!. In other words, a world specifies an exact instantiation of all the variables
in $V$.

A \emph{$V$-literal} is an ordered pair of the form $(X_i, x)$, where 
$X_i \in V$ and $x \in val(X_i)$. $V$-literals are similar to valuations, but the syntactic difference is introduced to
emphasise their difference in use, as will become clear later. $V$-literals are a generalisation of Boolean literals,
which are typically denoted as $x$ or $\neg x$ for the Boolean random 
variable $X$. For simplicity, and following this connection, in this paper we will often use the 
notation $X$ for $(X, T)$ and $\neg X$ for $(X, F)$.
A \emph{$V$\mbox{-}context} is any set of $V$-literals. A $V$\mbox{-}context is \emph{consistent} if it contains at most one literal for
each random variable. We will often also call $V$-contexts \emph{primitive contexts}.

A \emph{Bayesian network} is defined as a pair $\bnDL = (G,\Theta)$ where $G = (V, E)$ is a directed acyclic graph (DAG) 
and $\Theta$ is a set of conditional probability distributions for every variable $X\in V$ given its parents $\pi(X)$
on the DAG $G$; more precisely, this set has the form
$\Theta = \{P(X=x|\pi(X)=\vec{x'}) \mid X\in V\}.$

The BN \Bmc specifies a full JPD for the variables in $V$ by considering independence assumptions depicted in the graph
$G$; namely, every variable is (conditionally) independent of all its \emph{non-descendants} given its parents. Under this
assumption, it is easy to see that the JPD of $V$ can be computed through the chain rule 
$$P(\vec{X}=\vec{x})=\prod_{X_i\in V}P(X_i=x_i\mid \pi(X_i)=\vec{x_j});$$
that is, the probability of a world is obtained by multiplying the conditional probabilities of the valuations found in the tables. \modify{We let $P_\Bmc$ denote the JPD defined by the BN \Bmc.}

\begin{example}
Figure~\ref{fig:BN} depicts a BN with four random variables denoting the likelihood of different characteristics
of a construction: $X$ stands for a post-1986 building, $Y$ for a renovated building, $Z$ for the presence
of lead pipes, and $W$ for the safety of drinking water. 
\begin{figure}[t]
\centering
\includegraphics{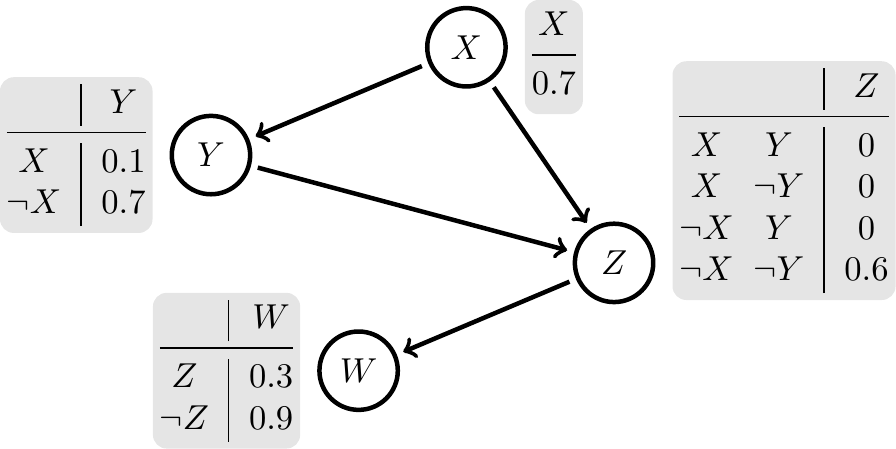}
\caption{A Bayesian network with four Boolean variables.}
\label{fig:BN}
\end{figure}
The tables attached to each node are the conditional probability distribution of that node given its parents. 
Hence, the BN expresses that a post-1986 building has only probability $0.1$ of being renovated.

The DAG expresses the conditional independence of $W$ and $Y$ given $Z$. That is, once that we observe the
value of $Z$ (if we know whether a house has lead pipes or not) then the probability of $W$ is not affected by
the knowledge of the renovation status of a house.

Through the chain rule, we can derive e.g., 
\begin{align*}
P(\neg X,\neg Y, Z, \neg W)={} & 0.3\cdot 0.3\cdot 0.6 \cdot 0.7 = 0.0378, & \text{and}\\
P(X,Y,Z,W)={} & 0.7\cdot 0.1\cdot 0\cdot 0.3=0.
\end{align*}
That is, it is very unlikely to find an old, non-renovated house, with lead pipes, and unsafe water; and 
renovated post-1986 houses with safe water cannot have lead pipes.
Note that to express the full JPD of these four variables directly, we would need a table with 16 rows.
\end{example}

\subsection{The Description Logic \ALC} 
\ALC is the smallest propositionally closed DL \cite{BCea-07,ScSc-91}. 
As with all DLs, its fundamental notions are those of
\emph{concepts} which correspond to unary predicates of first-order logic, and \emph{roles} corresponding to
binary predicates. Formally, given the mutually disjoint sets $N_I, N_C$, and $N_R$ of \emph{individual}, 
\emph{concept}, and \emph{role names}, respectively, the class of \ALC \emph{concepts} is built through the grammar rule
$$C::=A\mid\neg C\mid C\sqcap C\mid C\sqcup C\mid \exists r.C\mid \forall r.C,$$ 
where $A\in N_C$ and $r\in N_R$. 

The knowledge from an application domain is represented through a set of axioms, which restrict the possible
interpretations of concepts and roles. 
In \ALC, \emph{axioms} are either \emph{general concept inclusions} (GCIs) of the form $C\sqsubseteq D$,
\emph{concept assertions} $C(a)$, or \emph{role assertions} $r(a,b)$ where $a,b\in N_I, r\in N_R$, and $C,D$
are concepts.
An \emph{ontology} is a finite set of axioms. As customary in DLs, we sometimes partition an ontology into the \emph{TBox} 
\Tmc composed exclusively of GCIs, and the \emph{ABox} \Amc containing all concept and role assertions,
when it is relevant which kind of axiom is being used.

The semantics of \ALC is defined by interpretations akin to the first-order semantics. An \emph{interpretation} is a pair of the form 
$\Imc=(\Delta^\Imc,\cdot^\Imc)$
where $\Delta^\Imc$ is a non-empty set called the \emph{domain} and $\cdot^\Imc$ is the \emph{interpretation 
function} that maps every $a\in N_I$ to an element $a^\Imc\in\Delta^\Imc$, every $A\in N_C$ to a set
$A^\Imc\subseteq\Delta^\Imc$ and every $r\in N_R$ to a binary relation 
$r^\Imc\subseteq\Delta^\Imc{\times}\Delta^\Imc$. 
This interpretation function is inductively extended to arbitrary concepts by defining for any two concepts $C,D$:
\begin{itemize}
\item $(\neg C)^\Imc:=\Delta^\Imc\setminus C^\Imc$,
\item $(C\sqcap D)^\Imc:=C^\Imc\cap D^\Imc$,
\item $(C\sqcup D)^\Imc:=C^\Imc\cup D^\Imc$,
\item $(\exists r.C)^\Imc:=\{\delta\in\Delta^\Imc\mid \exists \eta\in C^\Imc. (\delta,\eta)\in r^\Imc\}$, and
\item $(\forall r.C)^\Imc:=\{\delta\in\Delta^\Imc\mid \forall \eta\in \Delta^\Imc. (\delta,\eta)\in r^\Imc\Rightarrow \eta\in C^\Imc\}$.
\end{itemize}
The interpretation \Imc \emph{satisfies} the GCI $C\sqsubseteq D$
iff $C^\Imc\subseteq D^\Imc$; the concept assertion $C(a)$ iff $a^\Imc\in C^\Imc$; and the role assertion
$r(a,b)$ iff $(a^\Imc,b^\Imc)\in r^\Imc$. We denote as $\Imc\models\alpha$ if \Imc satisfies the axiom $\alpha$. 
\Imc is a \emph{model} of the ontology \Omc (denoted by $\Imc\models\Omc$) iff it satisfies all axioms in
\Omc.

An important abbreviation in \ALC is the \emph{bottom} concept $\bot:=A\sqcap \neg A$, where $A$ is any concept name.
Clearly, this concept stands for a contradiction and, for any interpretation \Imc, $\bot^\Imc=\emptyset$.
Similarly, the \emph{top} concept $\top:=A\sqcup \neg A$ stands for a tautology and $\top^\Imc=\Delta^\Imc$ holds for
every interpretation \Imc.

The basic reasoning task in \ALC consists in deciding whether a given ontology \Omc is \emph{consistent}; that is, whether
there exists a model \Imc of \Omc. This problem is known to be \ExpTime-complete \cite{Schi-91,DoMa-00}. Other reasoning
problems, such as subsumption and instance checking, can be polynomially reduced to this one and hence preserve the
same complexity.

\begin{example}
\modify{
In \ALC it is possible to express the notion that water pipes do not contain lead through the GCI
$\et{WaterPipe}\sqsubseteq \forall \et{contains}.\neg\et{Lead}$. In addition, we can express that the object
\et{pipe1} is a water pipe
($\et{WaterPipe}(\et{pipe1})$); that \et{pipe1} contains \et{substance1} by the assertion 
$\et{contains}(\et{pipe1},\et{substance1})$; and that
\et{substance1} is in fact lead ($\et{Lead}(\et{substance1})$).}
Note that the ontology containing all four axioms is inconsistent.
\end{example}

\section{\BALC}
\label{sec:BALC}

We now introduce the probabilistic DL \BALC, which combines BNs, to compactly express joint probability distributions, and
\ALC to express background knowledge. In addition, \BALC can express \emph{logical} (as opposed to probabilistic) dependencies
between axioms. More precisely, \BALC axioms are required to hold only in some (possibly uncertain) contexts, which are 
expressed through annotations. The uncertainty of these contexts is expressed by the BN. These ideas are formalised next.

\begin{definition}[KB]
Let $V$ be a finite set of discrete random variables. 
A \emph{$V$\mbox{-}restricted axiom} ($V$-axiom) is an expression of the form 
$\alpha^\kappa$, where $\alpha$ is an \ALC axiom and $\kappa$ is a $V$-context. 
A \emph{$V$-ontology} is a finite set of $V$-axioms.
A $\BALC$ \emph{knowledge base} (KB) over $V$ is a pair $\kbDL = (\Omc,\Bmc)$ where 
$\Bmc$ is a BN over $V$\!, and $\Omc$ is a $V$\mbox{-}ontology.
\end{definition}
Note that the contexts labelling the axioms, and the variables of the BN from a \BALC KB both come from the same set of
discrete random variables $V$\!.

To define the semantics of \BALC, we extend the notion of an interpretation from \ALC to also take contexts into account.
The probabilities expressed by the BN are associated to the axioms following the multiple-world approach.

\begin{definition}[$V$-interpretation]
Let $V$ be a finite set of discrete random variables.
A \emph{$V$-interpretation} is a tuple of the form $\vModelDL = (\vDomainDL, \vInterpFuncDL, \valFuncDL)$ where 
$(\Delta^\Vmc,\cdot^\Vmc)$ is an \ALC interpretation and
$\valFuncDL$ is a \emph{valuation function} $\valFuncDL : V \to \cup_{X \in V} val(X)$ which maps every $X\in V$ to some
$\valFuncDL(X)\in val(X)$.
\end{definition}
Given a valuation function $\valFuncDL$, a Bayesian world $\omega$, and a context $\kappa$ we
use the notation $\valFuncDL = \omega$ to express that $\valFuncDL$ assigns to each random variable the same value as it 
has in $\omega$; $\valFuncDL \models \kappa$ to state that $\valFuncDL(X) = x$ for all $(X, x) \in \kappa$; 
and $\omega \models \kappa$ when there is $\omega = \valFuncDL$ such that $\valFuncDL \models \kappa$.

\begin{definition}[Model]
\label{def:model}
The $V$-interpretation $\vModelDL=(\vDomainDL, \vInterpFuncDL, \valFuncDL)$ is a \emph{model} of the $V$-axiom $\alpha^\kappa$ 
(denoted as $\vModelDL \models \alpha^\kappa$) iff (i) $\valFuncDL \not \models \kappa$, or 
(ii) $(\Delta^\Vmc,\cdot^\Vmc)$ satisfies $\alpha$. 
\Vmc is a model of the ontology $\Omc$ iff it is a model of all axioms in \Omc. In this case, we denote it by $\Vmc\models\Omc$.
\end{definition}
It is important to notice at this point that \BALC is a generalisation of the classical DL \ALC. Recall that a $V$-context is a
set of $V$-literals; that is, of pairs $(X,x)$. In particular, the empty set of literals $\emptyset$ is also a $V$-context.
By definition, every valuation function $\valFuncDL$ is such that $\valFuncDL\models \emptyset$. This means that the
\ALC axiom $\alpha$ is equivalent to the $V$-axiom $\alpha^\emptyset$. As a whole every \ALC ontology can be seen
as a $V$-ontology by simply associating the empty context to all its axioms. The original \ALC ontology and the $V$-ontology
constructed this way have the same class of models, except that the latter allows for any arbitrary valuation function in its third
component.
For brevity, and the reasons just described, for the rest of this paper we will abbreviate axioms of the
form $\alpha^\emptyset$ simply as $\alpha$. 
When it is clear from the context, we will also omit the $V$ prefix and refer only to e.g.\ contexts, GCIs, or ontologies.

A $V$-interpretation---more precisely, its valuation function---refers to only one possible world of the random variables in $V$.
\BALC KBs, on the other hand, have information about the uncertainty of being in
one world or another. In the multiple-world semantics, probabilistic interpretations combine multiple $V$-interpretations and 
the probability distribution described by the BN to give information about the uncertainty of the axioms, and their consequences.

\begin{definition}[Probabilistic model]
A \emph{probabilistic interpretation} is a pair of the form $\mathcal{P = (J, P_J)}$, where $\mathcal{J}$ is a finite set of 
$V$-interpretations and $\mathcal{P_J}$ is a probability distribution over $\mathcal{J}$ such that 
$\mathcal{P_J(\vModelDL)} > 0$ for all $\mathcal{\vModelDL \in J}$.
The probabilistic interpretation $\probModelDL$ is a \emph{model} of the axiom
$\alpha^\kappa$ (denoted by $\probModelDL \models \alpha^{\kappa}$) iff every $\mathcal{\vModelDL \in J}$ is a model 
of $\alpha^{\kappa}$. $\probModelDL$ is a \emph{model} of the ontology \Omc iff 
every $\mathcal{\vModelDL \in J}$ is a model of $\Omc$.

We say that the probability distribution $\mathcal{P_J}$ is \emph{consistent} with the BN $\bnDL$ if for every possible world 
$\omega$ of the variables in $V$ it holds that 
$$\sum_{\mathcal{\vModelDL \in J}, \valFuncDL=\omega} P_{\Jmc}(\vModelDL) = P_\Bmc(\omega)$$ \modify{Recall that} 
$P_\Bmc$ denotes the joint probability distribution defined by the BN \Bmc.
The probabilistic interpretation $\probModelDL$ is a \emph{model} of the KB 
$\kbDL = (\Omc, \bnDL)$ iff it is a (probabilistic) model of \Omc, and is consistent with $\bnDL$.
\end{definition}
In other words, a probabilistic interpretation describes a class of possible worlds, in which the knowledge of \Omc is interpreted,
and each world is associated with a probability. For this interpretation to be a model of the KB, it has to satisfy the constraints
it requires; that is, the ontological knowledge must be satisfied in each world, and the probabilities should be coherent
with the distribution from the BN.

\begin{example}
\label{exa:ont}
Consider the KB $\Kmc=(\Omc,\Bmc)$ where \Bmc is the BN depicted in Figure~\ref{fig:BN}, and \Omc
is the ontology
\begin{align*}
\Omc := \{ & &
	\text{Pipe}\sqcap \exists\text{contains}.\text{Lead} \sqsubseteq \text{LeadPipe}^\emptyset, \\ & 
	\text{Pipe}\sqsubseteq \forall \text{contains}.\neg\text{Lead}^X, &
	\text{Pipe}\sqsubseteq \forall \text{contains}.\neg\text{Lead}^Y, \\ &
	\text{Pipe}\sqsubseteq \exists \text{contains}.\text{Lead}^Z, &
	\text{Water}\sqcap \exists \text{hasAlkalinity}.\text{Low} \sqsubseteq \neg \text{Drinkable}^Z, \\ &
	\text{Water} \sqsubseteq \text{Drinkable}^{W}, &
	\text{Water} \sqsubseteq \neg \text{Drinkable}^{\neg W} &
\}.
\end{align*}
The first axiom provides a partial definition, stating that every pipe which contains lead is a lead pipe. Note that
it is labelled with the empty context, which means that it always holds.
The two axioms in the second row express that pipes in post-1986 (that is, within the context $X$) and in renovated buildings 
(in context $Y$) do not contain lead. 

The axioms in the third row refer exclusively to the context of lead pipes ($Z$). In this
case, our knowledge is that pipes do contain lead, and that water with low alkalinity is not drinkable, as it 
absorbs the lead from the pipes it travels on. Notice that the axioms in the second row contradict 
$\text{Pipe}\sqsubseteq\exists\text{contains}.\text{Lead}$ appearing in the third row. This is
not a problem because they are required to hold in different contexts. Indeed, we can observe from Figure~\ref{fig:BN}
that any context which makes $Z$, and either $X$ or $Y$ true must have probability 0. This means that no probabilistic model
can contain a world that satisfies these conditions.

Finally, the last two axioms state that water at a building is drinkable iff we are in the context of drinkable water ($W$).
Note that it is important to provide both axioms in order to guarantee an equivalence. Suppose that we had not included the
last axiom (i.e., $\text{Water}\sqsubseteq\neg\text{Drinkable}^{\neg W}$). Then, by our semantics, we could still produce a 
$V$\mbox{-}interpretation \Vmc whose valuation satisfies $\neg W$ but such that 
$\text{Water}^\Vmc\cap\text{Drinkable}^\Vmc\not=\emptyset$. This can only be avoided by the introduction of the last axiom.

It is a simple exercise to verify that this ontology has a model. The smallest probabilistic model of the KB \Kmc contains 10
$V$-interpretations; one for each world with positive probability.
\end{example}
Recall that $V$-contexts are also called primitive contexts. This is to emphasise the difference with what we will,
from now on, call complex contexts.
Formally, a \emph{complex context} $\phi$ is a finite non-empty set of primitive contexts. Note that primitive contexts
can be seen as complex ones; 
e.g., the primitive context $\kappa$ corresponds to the complex context $\{ \kappa \}$, and we will use them
interchangeably when there is no risk of ambiguity.

Given a valuation function $\valFuncDL$ and a complex context $\phi = \{ \alpha_1 , \dots , \alpha_n \}$ we say 
that $\valFuncDL$ satisfies $\phi$ (written as $\valFuncDL \models \phi$) iff $\valFuncDL$ satisfies \emph{at least one}
$\alpha_i \in \phi$; in particular, $\valFuncDL \models \kappa$ iff $\valFuncDL \models \{\kappa\}$. Thus, for the rest of the paper 
we assume that all contexts are in complex form unless 
explicitly stated otherwise. 
We emphasise here that for a valuation function to satisfy a complex context, it suffices that it satisfies one of the
simple contexts in it. Intuitively, a simple context is a conjunction of valuations, while a complex context is a disjunction
of simple contexts. However, one should not forget that the variables appearing in these contexts can potentially take
more than two different \modify{values}.

\begin{definition}
Given complex contexts $\phi = \{ \kappa_1, \dots , \kappa_n \}$ and $\psi = \{ \lambda_1, \dots, \lambda_m \}$ we define the operations
\begin{align*}
\phi \lor \psi & {} := \phi \cup \psi, & \text{and} \\
\phi \land \psi & {} := \bigcup_{\kappa \in \phi, \lambda \in \psi} \{ \kappa \cup \lambda \} = 
	\{ \kappa \cup \lambda \mid \kappa \in \phi, \lambda \in \psi \}.
\end{align*}
\end{definition}
These operations generalise propositional disjunction ($\lor$) and propositional conjunction 
($\land$), where disjunction has the property that either one of the two contexts holds and conjunction requires 
that both hold. 

Given two complex contexts $\phi$ and $\psi$, we say that $\phi$ entails $\psi$ (denoted by $\phi \models \psi$) 
iff for all $\valFuncDL$ such that 
$\valFuncDL \models \phi$ it follows that
$\valFuncDL \models \psi$. Alternatively $\phi \models \psi$ holds iff for all Bayesian worlds $\omega$ such that 
$\omega \models \phi$ it follows that $\omega \models \psi$.
It is easy to see that for all worlds $\omega$ and complex contexts $\phi,\psi$ it holds that
(i) $\omega \models \phi \lor \psi$  iff $\omega \models \phi$ or $\omega \models \psi $,
and (ii)
$\omega \models \phi \land \psi$ iff $\omega \models \phi$ and $\omega \models \psi$.
Two important special complex contexts are top ($\top$) and bottom ($\perp$), which are satisfied by 
all or no world, respectively.
If there are $n$ consistent primitive contexts $\kappa_1,\ldots,\kappa_n$ and $\kappa$ is an inconsistent context, 
these are defined as 
$\top:=\{ \kappa_1, \dots , \kappa_n \}$ and $\perp:=\kappa$.

This concludes the definition of the relevant components of the logic \BALC. 
In the next section, we study the problem of consistency of a \BALC KB, and its relation to other reasoning
problems.

\section{Consistency}
\label{sec:consistency}

As for \ALC, the most basic decision problem one can consider in \BALC is \emph{consistency}. Formally this
problem consists of deciding whether a 
given \BALC KB \Kmc has a probabilistic model or not. To deal with this problem it is convenient to consider
the classical \ALC ontologies that should hold at each specific world, which we call the restriction. 

\begin{definition}[restriction]
Given the \BALC KB
$\Kmc=(\Omc,\Bmc)$ and the world $\omega$, the \emph{restriction} of \Omc to $\omega$ is the \ALC ontology
defined by
\[
\Omc_\omega := \{ \alpha \mid \alpha^\kappa \in \Omc, \omega\models \kappa \}.
\]
\end{definition}
Recall that a probabilistic model $\Pmc=(\Jmc,\Pmc_\Jmc)$ of $\Kmc=(\Omc,\Bmc)$ is a class of classical interpretations 
associated to worlds $(\Delta^\Vmc,\cdot^\Vmc,\omega)$, where each is a model of the $V$-ontology \Omc. For such
an interpretation to be a model of \Omc, it must hold that, for each axiom $\alpha^\kappa\in\Omc$ such that $\omega\models\kappa$,
$(\Delta^\Vmc,\cdot^\Vmc)\models\alpha$ (see Definition~\ref{def:model}). This condition is equivalent to stating that
$(\Delta^\Vmc,\cdot^\Vmc)$ is a classical model of $\Omc_\omega$. 
\begin{example}
Consider the KB \Kmc from Example~\ref{exa:ont}, and the two worlds $\omega_1=\{X,\neg Y,\neg Z,W\}$ and
\modify{$\omega_2=\{X,\neg Y,Z,W\}$}. We then have
\begin{align*}
\Omc_{\omega_1} := \{ &
			\text{Pipe}\sqcap \exists\text{contains}.\text{Lead} \sqsubseteq \text{LeadPipe}, 
			\text{Pipe}\sqsubseteq \forall \text{contains}.\neg\text{Lead}, \\ & 
			\text{Water} \sqsubseteq \text{Drinkable}\}, & \text{and} \\
\Omc_{\omega_2} := \{ &
			\text{Pipe}\sqcap \exists\text{contains}.\text{Lead} \sqsubseteq \text{LeadPipe}, 
			\text{Pipe}\sqsubseteq \forall \text{contains}.\neg\text{Lead}, \\ & 
			\text{Pipe}\sqsubseteq \exists \text{contains}.\text{Lead},
			\text{Water}\sqcap \exists \text{hasAlkalinity}.\text{Low} \sqsubseteq \neg \text{Drinkable}, \\ &
			\text{Water} \sqsubseteq \text{Drinkable}\}.
\end{align*}
Note that both ontologies are consistent, but $\Omc_{\omega_2}$ can only be satisfied by interpretations \Imc such that
$\text{Pipe}^\Imc=\emptyset$.
\end{example}
In addition to satisfying the restrictions, the probability distribution $\Pmc_\Jmc$ of a model must be consistent with the BN \Bmc. 
This means that the 
probabilities of the interpretations associated with the world $\omega$ must add to $P_\Bmc(\omega)$. 
Using this insight, we obtain the following result.
\begin{theorem} 
\label{thm:cons}
The \BALC KB $\Kmc=(\Omc,\Bmc)$ is consistent iff for every world $\omega$ with $P_\Bmc(\omega)>0$ it holds that 
$\Omc_\omega$ is consistent.
\end{theorem}
\begin{proof}
Suppose first that for every world $\omega$ such that $P_\Bmc(\omega)>0$ the ontology $\Omc_\omega$ is consistent.
We build a probabilistic model \Pmc of \Kmc as follows. Given a world $\omega$ with $P_\Bmc(\omega)>0$, let 
$\Imc_\omega=(\Delta^{\Imc_\omega},\cdot^{\Imc_\omega})$ be an arbitrary (classical) model of $\Omc_\omega$. We construct
the $V$\mbox{-}interpretation $\Vmc_\omega=(\Delta^{\Imc_\omega},\cdot^{\Imc_\omega},\omega)$; that is, the same interpretation
$\Imc_\omega$, but associated to the world $\omega$. We claim that $\Vmc_\omega\models\Omc$. Indeed, 
take an arbitrary $\alpha^\kappa\in\Omc$. If $\omega\not\models\kappa$, then $\Vmc_\omega\models \alpha^\kappa$ trivially.
Otherwise, $\alpha\in\Omc_\omega$ and hence (since $\Imc_\omega\models\Omc_\omega$) 
$(\Delta^{\Imc_\omega},\cdot^{\Imc_\omega})\models\alpha$. Let now $\Omega$ be the set of all worlds $\omega$ such that
$P_\Bmc(\omega)>0$. We define the set of $V$-interpretations $\Jmc_\Omega:=\{\Vmc_\omega\mid\omega\in\Omega\}$, and
the probability distribution $\Pmc_{\Jmc_\Omega}$ that sets $\Pmc_{\Jmc_\Omega}(\Vmc_\omega)=P_\Bmc(\omega)$ for every 
$\omega\in\Omega$.
The probabilistic interpretation $\Pmc=(\Jmc_\Omega,\Pmc_{\Jmc_\Omega})$ is a model of \Kmc.

Conversely, suppose that \Kmc is consistent, and let $\Pmc=(\Jmc,\Pmc_\Jmc)$ be a model of \Kmc. Given a world $\omega$
such that $P_\Bmc(\omega)>0$, we know that 
$\sum_{\mathcal{\vModelDL \in J}, \valFuncDL=\omega} P_{\Jmc}(\vModelDL) = P_\Bmc(\omega)$. In particular this means
that there must exist some $\Vmc=(\Delta^\Vmc,\cdot^\Vmc,\valFuncDL)\in\Jmc$ such that $\valFuncDL=\omega$. 
$\Pmc\models\Kmc$ means that $\Vmc\models\Omc$ and hence $(\Delta^\Vmc,\cdot^\Vmc)\models\Omc_\omega$. Thus,
$\Omc_\omega$ is consistent.
\end{proof}
Based on this result, we can derive a process for deciding consistency that provides a tight complexity bound
for this problem.
\begin{corollary}
\label{cor:et}
\BALC KB consistency is \ExpTime-complete.
\end{corollary}
\begin{proof}
Recall that consistency checking in \ALC is \ExpTime-complete, and that \BALC is a generalisation of \ALC. This
yields \ExpTime-hardness for \BALC consistency. To obtain the upper bound, note that
there are exponentially many worlds $\omega$. For each of them, we can check (classical) consistency
of $\Omc_\omega$ (in exponential time) and that $P_\Bmc(\omega)>0$, which is linear in the size of \Bmc by the chain rule.
\end{proof}
The algorithm described in the proof of this corollary is optimal in terms of \emph{worst-case} complexity, but
it also runs in exponential time in the \emph{best case}. Indeed, it enumerates all the (exponentially many)
Bayesian worlds, before being able to make any decision. In practice, it is infeasible to use an algorithm that requires 
exponential time on every instance.
For that reason, we present a new algorithm based on the tableau method originally developed for \ALC. To 
describe this algorithm, we need to introduce some additional notation.

We denote the context that describes all worlds $\omega$ such that $\Omc_\omega$ is inconsistent as 
$\phi_{\kbDL}^{\bot}$. That is, the context $\phi_\Kmc^\bot$ is such that 
$\omega \models \phi_{\kbDL}^\bot$ iff $\Omc_\omega$ is inconsistent.
Moreover, $\phi_{\bnDL}$ is a context such that $\omega \models \phi_{\bnDL}$ iff $P(\omega)=0$.
Theorem~\ref{thm:cons} states that \Kmc is inconsistent whenever there is a world that models 
both $\phi_{\kbDL}^\bot$ and $\neg\phi_{\bnDL}$. This is formalized in the following result.

\begin{theorem}
\label{thm:formula}
Given the \BALC KB $\Kmc=(\Omc,\Bmc)$, let $\phi_\Kmc^\bot$ and $\phi_\Bmc$ be the contexts described above.
Then \Kmc is inconsistent iff $\phi_\Kmc^\bot \land \neg\phi_\Bmc$ is satisfiable. 
\end{theorem}
\begin{proof}
\Kmc is inconsistent iff there is a world $\omega$ such that $P_\Bmc(\omega)>0$ and $\Omc_\omega$ is inconsistent
(Theorem~\ref{thm:cons}) iff $\omega\not\models\phi_\Bmc$ and $\omega\models\phi_\Kmc^\bot$ iff 
$\omega\models\neg\phi_\Bmc\land\phi_\Kmc^\bot$ iff $\neg\phi_\Bmc\land\phi_\Kmc^\bot$ is satisfiable.
\end{proof}
\begin{example}
Consider once again the KB $\Kmc=(\Omc,\Bmc)$ from Example~\ref{exa:ont}, and notice that $\Omc_\omega$ is consistent for all
worlds $\omega$. Hence $\phi_\Kmc^\bot=\bot$. In particular, this means that $\phi_\Kmc^\bot\land\neg\phi_\Bmc$ is
unsatisfiable, and hence \Kmc is consistent.

Let now $\Kmc'=(\Omc',\Bmc)$, where $\Omc'=\Omc\cup\{\text{Pipe}(\text{pipe1})^\emptyset\}$. 
We can see that for any world $\omega$ which entails any of the contexts $X,Z$ or $Y,Z$, $\Omc'_\omega$ is inconsistent.
Hence we have that $\phi_{\Kmc'}^\bot=\{\{X,Z\},\{Y,Z\}\}$. Looking at the probability tables from \Bmc, we also see that
$\phi_\Bmc=\{\{X,Z\},\{Y,Z\}\}$; that is, any world that satisfies $\{X,Z\}$ or $\{Y,Z\}$ must have probability 0. Thus,
\modify{$\phi_{\Kmc'}^\bot\land\lnot\phi_\Bmc=\bot$} and $\Kmc'$ is consistent as well.
\end{example}
To decide consistency, it then suffices to find a method which is capable of deriving the contexts $\phi_{\kbDL}^\bot$ and 
$\phi_{\bnDL}$. For the former, we present a variant of the glass-box approach for so-called axiom 
pinpointing~\cite{LMPB06,MLBP06,BaPe-JLC09}, which is originally based on the ideas from~\cite{BaHo-95}.
This approach modifies the standard tableaux algorithm for $\ALC$---which tries to construct a model by decomposing
complex constraints into simpler ones---to additionally keep track of the
contexts in which the derived elements from the tableau hold. In a nutshell, tableau algorithms apply different decomposition
rules to make implicit constraints explicit. For axiom pinpointing, whenever a rule application requires
the use of an axiom from the ontology, this fact is registered as part of a propositional formula. In our case,
we need a context, rather than a propositional formula, to take care of the multiple values that the random
variables take. Although the underlying idea remains essentially unchanged, the handling of these multiple values
requires additional technicalities. Before we look at this algorithm in detail, it is important to keep in mind that the
goal is to detect all the worlds whose restriction is an inconsistent \ALC ontology; hence, we are interested in detecting
and highlighting contradictions.

For this algorithm, it will be important to distinguish the terminological portion of the ontology (the TBox) from the 
assertional knowledge (the ABox). From a very high-level view, the algorithm starts by creating a partial model that
satisfies all the constraints from the ABox, and then uses rules to decompose complex concepts until an actual
valuation (based exclusively on concept names and role names) is reached. Throughout the process, rules are 
also used to guarantee that all the constraints from the TBox are satisfied by the individuals in this model.

The algorithm starts with the ABox \Amc from \Omc. If \Omc contains no assertions, we assume implicitly that the ABox is
$\{\top(a)\}$ for an arbitrary individual name $a$. Recall that all the axioms in \Amc are labeled with a context, and
we preserve the context at this initialisation step. From this point, the algorithm deals with a set of ABoxes \Amf, which are
created and modified following the rules from Figure~\ref{fig:rules}. 
\begin{figure}[tb]
\modify{
\begin{tabular}{@{}l r p{11cm}@{}}
\toprule
    $\sqcap$-rule & if & $\balcStmnt{(C_1 \sqcap C_2)(a)}{\phi} \in \mathcal{A}$, and
     either $\balcStmnt{C_1(a)}{\phi}$ or $\balcStmnt{C_2(a)}{\phi}$ is $\mathcal{A}$-insertable \\
    & then & $\aboxDL^\prime := (\aboxDL \oplus \balcStmnt{C_1(a)}{\phi})\oplus \balcStmnt{C_2(a)}{\phi}$ \\
    $\sqcup$-rule & if & $\balcStmnt{(C_1 \sqcup C_2)(x)}{\phi} \in \aboxDL$, and
     both $\balcStmnt{C_1(a)}{\phi}$ and $\balcStmnt{C_2(a)}{\phi}$ are $\mathcal{A}$-insertable \\
    & then & $\aboxDL^\prime := \aboxDL \oplus \balcStmnt{C_1(a)}{\phi}$, $\aboxDL^{\prime \prime} := \aboxDL \oplus \balcStmnt{C_2(a)}{\phi}$\\
    $\exists$-rule & if & $\balcStmnt{(\exists R.C)(a)}{\phi} \in \aboxDL$,
     there is no $c$ such that neither $\balcStmnt{R(a,c)}{\phi}$ nor $\balcStmnt{C(c)}{\phi}$ is 
     $\mathcal{A}$\mbox{-}insertable, and $a$ is not blocked \\
    & then & $\aboxDL^\prime := (\aboxDL \oplus \balcStmnt{R(a,b)}{\phi}) \oplus \balcStmnt{C(b)}{\phi}$, where $b$ is a new individual name \\ 
    $\forall$-rule & if & $\{ \balcStmnt{(\forall R.C)(a)}{\phi}, \balcStmnt{R(a,b)}{\psi}\}\subseteq \aboxDL$, and
    $\balcStmnt{C(b)}{{\phi \wedge \psi}}$ is $\mathcal{A}$-insertable \\
    & then & $\aboxDL^\prime := \aboxDL \oplus \balcStmnt{C(b)}{{\phi \wedge \psi}}$\\
    $\sqsubseteq$-rule & if & $\balcStmnt{(C \sqsubseteq D)}{\phi} \in \Omc$, $a$ appears in $\aboxDL$, and
    $\balcStmnt{(\neg C \sqcup D)(a)}{{\phi}}$ is $\mathcal{A}$-insertable \\
    & then& $\aboxDL^\prime := \aboxDL \oplus \balcStmnt{(\neg C \sqcup D)(a)}{{\phi}}$ \\
\bottomrule
\end{tabular}
\caption{Expansion rules for constructing $\phi_{\kbDL}^\bot$}
}
\label{fig:rules}
\end{figure}
\modify{Throughout the execution of the algorithm, each assertion in the ABoxes is labelled with a context which, as usual, is
a propositional formula over the context variables; the algorithm modifies these labels accordingly to store all the relevant
information.}
As a pre-requisite for the execution of the algorithm, we assume that all concepts appearing in the ontology
are in \emph{negation normal form}~(NNF); that is, only concept names can appear in the scope of a negation 
operator. This assumption is made w.l.o.g.\ because every concept can be transformed into NNF in linear time by
applying the De Morgan laws, the duality of the quantifiers, and eliminating double negations.
Each rule application chooses an ABox $\Amc\in\Amf$ and replaces it by one or two new ABoxes that 
expand \Amc with new assertions. We explain the details of these rule applications next.

An assertion $\alpha^\phi$ is \emph{$\mathcal{A}$-insertable} 
iff for all $\psi$ 
such that $\alpha^{\psi} \in \mathcal{A}$, $\phi \not \models \psi$.
In the expansion rules, the symbol $\oplus$ is used as shorthand for the operation
\[
\aboxDL \oplus \alpha^{\phi} := 
	\begin{cases}
		(\aboxDL \setminus \{ \alpha^{\psi} \}) \cup \{ \alpha^{{\phi \vee \psi}} \} & \alpha^{\psi} \in \aboxDL \\
		\aboxDL \cup \{ \alpha^{\phi} \} & \text{otherwise}.
	\end{cases}
\]
\modify{The main idea of this operator is to update the contexts in which an assertion is required to hold. Either the
assertion is not already present in the ABox, and hence it is included with the context obtained by conjoining all the
contexts triggering the rule application or, if the assertion is already present, there are multiple ways to derive the same
assertion, whose contextual causes are all disjoint. Rather than keeping two different assertions in the ABox, we simply
update its label.}

Within an ABox \Amc, we say that
the individual \modify{$x$} is an \emph{ancestor} of \modify{$y$} if there is a chain of role assertions connecting \modify{$x$} to \modify{$y$}; more formally,
if there exist role names $r_1,\ldots,r_n$ and individuals \modify{$a_0,\ldots,a_n$} with $n\ge 1$ such that
\modify{$\{r_i(a_{i-1},a_i)\mid 1\le i\le n\}\subseteq\Amc$, $a_0=x$, and $a_n=y$}.
The node
\modify{$x$} \emph{blocks} \modify{$y$} iff \modify{$x$} is an ancestor of \modify{$y$} and for every \modify{$\balcStmnt{C(y)}{\psi} \in \aboxDL$}, there is
a $\phi$ such that \modify{$\balcStmnt{C(x)}{\phi} \in \aboxDL$} and $\psi \models \phi$; \modify{$y$} is \emph{blocked} if there
is a node that blocks it.

The algorithm applies the expansion rules from Figure~\ref{fig:rules} until \Amf is \emph{saturated}; i.e., 
until no rule is applicable to any
$\Amc\in\Amf$. 
We say that the ABox \Amc contains a \emph{clash} if \modify{$\{A(x)^\phi,\neg A(x)^\psi\}\subseteq\Amc$} for some individual \modify{$x$} and 
concept name $A$. Informally, a clash means that the ABox has an obvious contradiction, as it requires to interpret the element
\modify{$x$} to belong to the concept $A$ and its negation. We define the context 
\modify{
\[
\phi_\Amc:= \bigvee_{A(a)^\phi, \neg A(a)^\psi \in \Amc} \phi \land \psi,
\]
}%
which intuitively describes the contexts of all the clashes that appear in \Amc. When \Amf is saturated, we return the
context $\phi_\Kmc^\bot=\bigwedge_{\Amc\in\Amf} \phi_\Amc$ expressing the need of having clashes in every
ABox \Amc for inconsistency to follow. This is because each ABox in \Amf is a potential model, where the multiplicity
arises from the non-deterministic choice when decomposing disjunctions $\sqcup$. If all potential choices are obviously
contradictory, then no model can exist.
It is important for the correctness of the algorithm to notice that the definition of a clash does not impose any constraints 
on the contexts
$\phi$ and $\psi$ labelling the assertions \modify{$A(x)$ and $\neg A(x)$}, respectively. Indeed, \modify{$A(x)$ and 
$\neg A(x)$} could hold in contradictory contexts. In that case, the conjunction appearing in $\phi_\Amc$ would
not be affected; i.e., this clash will not provide any new information about inconsistency.
For example, if $\Amc=\{\text{WaterPipe}(\text{pipe1})^X,\neg\text{WaterPipe}(\text{pipe1})^{\neg X}\}$, then the 
context $\phi_\Amc$ is $X\land\neg X$. This context cannot be satisfied by any world, meaning that \Amc is not inconsistent.

Informally, the context $\phi_\Kmc^\bot$ corresponds to the \emph{clash formula} (also called pinpointing formula) for explaining
inconsistency of an \ALC ontology~\cite{LMPB06,BaPe-JLC09}. The main differences are that the variables 
appearing in
a context are not necessarily Boolean, but multi-valued, and that the axioms in \Omc are not labelled with 
unique variables, but rather with contexts. 
\modify{The correctness of the approach follows from the fact that each assertion $\alpha$ is labelled with a formula representing
the contexts in which $\alpha$ is necessarily true. Thus, every inconsistent context will be represented in $\phi_\Kmc^\bot$.
Conversely, saturation of the expansion rules guarantees that no other derivations of inconsistency as possible; otherwise, 
some rule would still be applicable.}

Notice that the expansion rules in 
Figure~\ref{fig:rules} generalise the expansion rules for \ALC, but may require new rule applications to guarantee
that all possible derivations of a clash are detected. 
\modify{Indeed, while the classical algorithm never applies a rule where the derived assertion already exists in the ABox, 
our modified algorithm may apply the rule if the label associated to the assertion does not already imply that which will be
included.}
As observed in~\cite{BaPe07,BaPe-JLC09}, one has to be
careful with claims of termination of the modified method, since the additional rule applications may in fact lead to
an infinite execution. \modify{Specifically, the modified algorithm obtained from a terminating tableau is not guaranteed to
terminate.} 
\modify{However, \cite{BaPe-JLC09,PenaPhD} devised some general conditions, which are sufficient to guarantee termination
of a method modified in this way. Importantly, those conditions are met whenever the original algorithm deals only with
unary and binary assertions, as is the case of the \ALC algorithm presented here. Hence, our approach is also guaranteed to
terminate. This yields the following result.}
\begin{theorem}
\label{thm:tab}
The modified tableau algorithm terminates, and the context $\phi_\Kmc^\bot$ is such that for every
world $\omega$, $\omega\models\phi_\Kmc^\bot$ iff $\Kmc_\omega$ is inconsistent.
\end{theorem}
The proof of this theorem is analogous to the proofs of correctness and termination of pinpointing extensions for
tableau algorithms from \cite{BaPe-JLC09} only taking the differences between propositional formulas and contexts
into account.

We now turn our attention to the computation of the formula $\phi_\Bmc$. Recall that in a BN, the joint 
probability distribution is the product of the conditional probabilities of each variable given its parents. Hence
a world $\omega$ can only have probability 0 if it evaluates some variable in $X\in V$ and its parents
$\pi(X)$ to values $x$ and $\vec{x}$, respectively, such that $P(X=x\mid \pi(X)=\vec{x})=0$. Thus, to compute
$\phi_\Bmc$ it suffices to find out the cells in the conditional probability tables in $\Theta$ with value $0$.

\begin{theorem}
\label{thm:bn}
Let $\Bmc=(V,\Theta)$ be a BN, and define 
\modify{
\[
\phi_\Bmc:=\bigvee_{P(X=x\mid \pi(X)=\vec{x})=0} \left((X,x)\land\bigwedge_{Y\in\pi(X),Y=y}(Y,y)\right).
\]
}
Then for every world $\omega$, $\omega\models\phi_\Bmc$ iff $P_\Bmc(\omega)=0$.
\end{theorem}
\modify{
\begin{proof}
If $\omega\models\phi_\Bmc$ then there is a variable $X$ such that $\omega\models (X,x)\land\bigwedge_{Y\in\pi(X),Y=y}(Y,y)$
and $P(X=x\mid \pi(X)=\vec{x})=0$. By the chain rule it follows that $P_\Bmc(\omega)=0$. Conversely, if
$0=P_\Bmc(\omega)=\prod_{X\in V}P(X\mid \pi(X))$ it must be the case that at least one factor in the product is 0; but then
that factor satisfies a disjunct in $\phi_\Bmc$.
\end{proof}
}
Notice that, in general, the context $\phi_\Bmc$ can be computed faster than simply enumerating all possible
worlds. In particular, if the conditional probability tables in $\Theta$ contain no 0-valued cell, then 
$\phi_\Bmc=\bot$; i.e., it is satisfied by no world.

To summarise, in this section we have shown that deciding inconsistency of a \BALC KB $\Kmc=(\Omc,\Bmc)$ can be reduced
to checking satisfiability of the context $\phi_\Kmc^\bot\land\neg\phi_\Bmc$, where $\phi_\Kmc^\bot$ is the context
representing all the worlds $\omega$ where $\Omc_\omega$ is inconsistent, and $\phi_\Bmc$ is the context representing
all worlds $\omega$ with $P_\Bmc(\omega)=0$. We then provided a tableaux-like method for computing $\phi_\Kmc^\bot$,
and an enumeration approach for finding $\phi_\Bmc$. We note that, just like in classical \ALC, the tableau algorithm is not
optimal in terms of computational complexity. \modify{In fact, although consistency of \ALC ontologies is in ExpTime
\cite{DoMa-00}, the tableau algorithm we use as the basis of our approach runs in double exponential time; the main reason
for this is the $\sqcup$-rule which may duplicate the number of ABoxes under consideration. We 
conjecture that it can be implemented in a goal-directed manner which behaves well in practice, but leave this for future work.}

Although consistency is a very important problem to be studied, we are interested also in other reasoning tasks,
which can arise through the use of probabilistic ontologies.
In particular, we should also take into account the contexts and the probabilities provided by the BN beyond the 
question of 
whether they are positive or not. In the next section we study variants of satisfiability and subsumption problems,
before turning our attention to instance checking.

\section{Satisfiability and Subsumption}

In this section, we focus on two problems---satisfiability and subsumption---which depend only on the TBox part of an ontology, 
and hence assume
for the sake of simplicity that the ABox is empty. Thus, we will often write a \BALC KB as a pair $(\Tmc,\Bmc)$ where
\Tmc is a TBox and \Bmc is a BN. We are in general interested in understanding the properties and relationships
of concepts. 

\begin{definition}
Given two concepts $C,D$ and a \BALC KB \Kmc, we say that $C$ is \emph{satisfiable} w.r.t.\ \Kmc iff 
there exists a probabilistic model 
$\mathcal{P = (J, P_J)}$ of $\kbDL$ s.t. $C^\vModelDL\not=\emptyset$ for all $\vModelDL\in \mathcal{J}$. We say that $C$ is
\emph{subsumed} by $D$ w.r.t.\ \Kmc iff for all models $\Pmc=(\Jmc,\Pmc_\Jmc)$ of \Kmc and all $\Vmc\in\Jmc$
$C^\Vmc\subseteq D^\Vmc$.
\end{definition}
 It is possible to adapt the well known reductions from the classical case to show
that these two problems are \ExpTime-complete.
\begin{theorem}
Satisfiability and subsumption w.r.t.\ \BALC KBs are \ExpTime-complete.
\end{theorem}
\begin{proof}
Let $\Kmc=(\Tmc,\Bmc)$ and $C,D$ be concepts.
It is easy to see that $C$ is subsumed by $D$ w.r.t.\ \Kmc iff the KB
$\Kmc'=(\Tmc\cup\{(C\sqcap \neg D)(a)^\emptyset\},\Bmc)$
is inconsistent, where $a$ is an arbitrary individual name. Similarly, $C$ is satisfiable iff 
$\Kmc''=(\Tmc\cup\{C(a)^\emptyset\},\Bmc)$ is consistent.
\end{proof}
In the following, we study variants of these problems which take contexts and probabilities into account. 
For a more concise presentation, we will present only
the cases for subsumption. Analogous results hold for satisfiability based on the fact that for every \ALC 
interpretation \Imc, it holds that $C^\Imc=\emptyset$ iff $C^\Imc\subseteq\bot^\Imc$; hence, 
\modify{an approach that is similar to the one used to solve subsumption can be used to solve satisfiability.}
First we consider additional information about contexts; afterwards we compute the probability of an entailment,
and finally the combination of both kinds of information.

\begin{definition}[contextual subsumption]
\label{def:con:sub}
Let $\Kmc=(\Tmc,\Bmc)$ be a \BALC KB, $C,D$ concepts, and $\kappa$ a context. $C$ is \emph{subsumed \modify{by $D$}
in context $\kappa$} w.r.t.\ \Kmc, denoted as $\Kmc\models (C\sqsubseteq D)^\kappa$\!, if every probabilistic
model of \Kmc is also a model of $(C\sqsubseteq D)^\kappa$\!.
\end{definition}
This definition introduces the natural extension of subsumption that considers also the contexts. \modify{Note that in our setting, contexts provide a means to express and reason with probabilities.}

\begin{definition}[subsumption probability]
\label{def:sub:probs}
Let $\probModelDL = (\mathcal{J}, P_{\mathcal{J}})$ be a probabilistic model of the KB $\kbDL$, $\kappa$ a
context, and
$C,D$ two concepts. The \emph{probability of $\balcStmnt{(C \sqsubseteq D)}{\kappa}$} w.r.t.\ \Pmc is
$$P_{\probModelDL}(\balcStmnt{(C \sqsubseteq D)}{\kappa})= 
	\sum_{\vModelDL \in \mathcal{J}, \vModelDL \models \balcStmnt{(C \sqsubseteq D)}{\kappa}} 
		P_{\mathcal{J}}(\vModelDL).$$
The \emph{probability of $\balcStmnt{(C \sqsubseteq D)}{\kappa}$} w.r.t.\  $\kbDL$ is
$$P_{\kbDL}(\balcStmnt{(C \sqsubseteq D)}{\kappa}) = 
	\inf_{\mathcal{P \models K}} P_{\mathcal{P}}(\balcStmnt{(C \sqsubseteq D)}{\kappa}).$$
We say that $C$ is \emph{positively subsumed} by $D$ in $\kappa$ iff $P_\Kmc((C\sqsubseteq D)^\kappa)>0$; it is
\emph{$p$-subsumed} iff $P_\Kmc((C\sqsubseteq D)^\kappa)\ge p$; it is \emph{exactly $p$-subsumed} iff
$P_\Kmc((C\sqsubseteq D)^\kappa)=p$; and it is \emph{almost certainly subsumed} iff
$P_\Kmc((C\sqsubseteq D)^\kappa)=1$.
\end{definition}
\modify{The definition above allows us to consider the problem of probabilistic reasoning.}
That is, the probability of a subsumption in a specific model is the sum of the probabilities of the worlds in which 
$C$ is subsumed by $D$ in context $\kappa$; notice that this trivially includes all worlds where $\kappa$ does not 
hold, by construction (see Definition~\ref{def:model}).
In the case where $\kbDL$ is inconsistent we define the probability of all subsumptions as $1$ to ensure our 
definition is consistent with general probability theory (recall that $\inf (\emptyset) = \infty$ in general). 

When we speak about the probability of a subsumption w.r.t.\ a \BALC KB, we consider the infimum over the probabilities
provided by all possible models of \Kmc. The reason for this is that we want to find the tightest constraint, taking into account
the open world semantics. The infimum is, to a degree, the natural interpretation of the information that holds in all possible
models. When reasoning about probabilities it is not always important to compute a precise value, but rather find some relevant
bounds. Thus, $p$-subsumption considers a lower bound, while positive subsumption only cares about the probability being zero
or not (dually, almost-certainty refers to the probability being one or not). All these decision problems have different applications
and properties.

Contextual subsumption is related to subsumption probability in the obvious way. Namely, a KB \Kmc entails a 
contextual subsumption iff the probability of the subsumption in \Kmc is 1. 

\begin{theorem}
Given a KB $\kbDL$, concepts $C$ and $D$, and a context $\kappa$, it holds that:
$$\kbDL \models \balcStmnt{(C \sqsubseteq D)}{\kappa} 
~\text{ iff }~ 
P_{\kbDL}(\balcStmnt{(C \sqsubseteq D)}{\kappa}) = 1.$$
\end{theorem}
\begin{proof}
$\Kmc\models(C\sqsubseteq D)^\kappa$ iff for every probabilistic interpretation $\Pmc$, if $\Pmc\models\Kmc$, then
$\Pmc\models(C\sqsubseteq D)^\kappa$ iff for every probabilistic model $\Pmc=(\Imc,P_\Imc)$ of \Kmc and every
$\Vmc\in\Imc$, $\Vmc\models(C\sqsubseteq D)^\kappa$ iff for every model \Pmc of \Kmc 
$P_\Pmc((C\sqsubseteq D)^\kappa)=1$ iff 
$\inf_{\mathcal{P \models K}} P_{\mathcal{P}}(\balcStmnt{(C \sqsubseteq D)}{\kappa})=1$.
\end{proof}
This is convenient as it provides a method of reusing our results from Section~\ref{sec:consistency}, which finds a context
describing inconsistency, to compute subsumption probabilities.

\begin{theorem}
\label{thm:sum}
Let $\kbDL = (\tboxDL, \bnDL)$ be a consistent \BALC KB, $C,D$ two concepts, and $\kappa$ a context.
For the KB $\Kmc'=(\Tmc\cup\{C(a)^\kappa,\neg D(a)^\kappa\},\Bmc)$ it holds that
\begin{equation*}
P_{\kbDL}(\balcStmnt{(C \sqsubseteq D)}{\kappa})  =  
	\sum_{\omega \models \phi_{\kbDL^\prime}^\bot} P_\Bmc(\omega) + 1 - P_\Bmc(\kappa).
\end{equation*}
\end{theorem}
\begin{proof}
Let $\Omc=(\Tmc\cup\{C(a)^\kappa,\neg D(a)^\kappa\})$, so $\Kmc'=(\Omc,\Bmc)$.
For an arbitrary but fixed model $\Pmc=(\Jmc,P_\Jmc)$ of \Kmc, we have that
\begin{align*}
P_\Pmc((C\sqsubseteq D)^\kappa)= {} & \sum_{\Vmc\in\Jmc,\Vmc\models(C\sqsubseteq D)^\kappa}P_\Jmc(\Vmc) \\
	= {} & \sum_{\valFuncDL\not\models\kappa}P_\Jmc(\Vmc) + 
		  \sum_{\valFuncDL\models\kappa,\Vmc\models C\sqsubseteq D} P_\Jmc(\Vmc) \\
	= {} & 1- P_\Bmc(\kappa) + 
		  \sum_{\Vmc\not\models \Kmc'} P_\Jmc(\Vmc) \\
	\le {} & 1- P_\Bmc(\kappa) + \modify{\sum_{\{\Vmc\mid\ \Omc_{\valFuncDL}\text{ inconsistent}\}} P_\Jmc(\Vmc)}.
\end{align*}
Hence 
$P_\Kmc((C\sqsubseteq D)^\kappa) \le \sum_{\omega \models \phi_{\kbDL^\prime}^\bot} P_\Bmc(\omega) + 1 - P_\Bmc(\kappa)$. 
For the lower bound, it suffices to build a minimal model of \Kmc, where for each world $\omega$ we have an interpretation
$\Vmc_\omega$ such that $\Vmc_\omega\models C\sqsubseteq D$ iff $\Omc_\omega\models C\sqsubseteq D$ as in the
proof of Theorem~\ref{thm:cons}.
\end{proof}
\begin{example}
\label{exa:prob:sub}
Returning to our running example with the KB \Kmc from Example~\ref{exa:ont}, suppose that we are interested in
finding the probability of water being drinkable. That is, we want to compute
$P_\Kmc((\text{Water}\sqsubseteq\text{Drinkable})^\emptyset)$. Since $P_\Bmc(\emptyset)=1$, Theorem~\ref{thm:sum}
states that it suffices to find the probabilities of all worlds $\omega$ such that the restriction of
$\Omc\cup\{\text{Water}(a),\neg\text{Drinkable}(a)\}$ to $\omega$ is inconsistent.
We note that the only such worlds are those which satisfy $W$. Hence, 
\begin{align*}
P_\Kmc((\text{Water}\sqsubseteq\text{Drinkable})^\emptyset) = {} &
	\sum_{\omega\models W}P_\Bmc(\omega).
\end{align*}
Those worlds, and their probabilities, are shown in Table~\ref{tab:worlds}.
\begin{table}
\caption{Worlds $\omega$ such that $\omega\models W$ and their probability w.r.t\ the BN from Figure~\ref{fig:BN}.}
\label{tab:worlds}
\begin{tabular}{@{}rrrr@{\qquad }l@{}}
\toprule
\multicolumn{4}{c}{world $\omega$} & $P_\Bmc(\omega)$ \\
\midrule
$X$ & $Y$ & $Z$ & $W$ & 0 \\
$X$ & $Y$ & $\neg Z$ & $W$ & 0.063\\
$X$ & $\neg Y$ & $Z$ & $W$ & 0 \\
$X$ & $\neg Y$ & $\neg Z$ & $W$ & 0.567\\
$\neg X$ & $Y$ & $Z$ & $W$ & 0 \\
$\neg X$ & $Y$ & $\neg Z$ & $W$ & 0.189\\
$\neg X$ & $\neg Y$ & $Z$ & $W$ & 0.0162\\
$\neg X$ & $\neg Y$ & $\neg Z$ & $W$ & 0.0108\\
\bottomrule
\end{tabular}
\end{table}
Adding all these values we get that the probability is $0.8460$. In other words, according to our model, the probability of a building
having drinkable water (in the absence of any other information) is very high.
\end{example}
Notice that the formula $\phi_{\Kmc'}^\bot$ requires at most exponential space on the size of \Tmc to be 
encoded. For each of the exponentially many worlds, computing $P_\Bmc(\omega)$ requires polynomial time
due to the chain rule. Hence, overall, the computation of the subsumption probabilities requires exponential time. 
Importantly, this bound does not depend on
how $\phi_{\Kmc'}^\bot$ was computed; it could have e.g.\ been computed through the process described in
Corollary~\ref{cor:et}. This provides an exponential upper bound for computing the 
probability of a subsumption.

\begin{corollary}
The probability of a subsumption w.r.t.\ a KB can be computed in exponential time 
on the size of the KB.
\end{corollary}
Obviously, an exponential-time upper bound for computing the exact probability of a subsumption relation
immediately yields an \ExpTime upper bound for deciding the other problems introduced in 
Definition~\ref{def:sub:probs} as well. All these problems are also generalisations of the subsumption problem in
\ALC. More precisely, given an \ALC TBox \Tmc, we can create the \BALC KB $\Kmc=(\Tmc',\Bmc)$ where
$\Tmc'$ contains all the axioms in \Tmc labelled with the context $x$ and \Bmc contains only one Boolean
node $x$ that holds with probability 1. Given two concepts $C,D$ $\Tmc\models C\sqsubseteq D$ iff
$C$ is almost certainly subsumed by $D$ in context $x$. Since subsumption in \ALC is already \ExpTime-hard,
we get that all these problems are \ExpTime-complete.

In practice, however, it may be too expensive to compute the exact probability when we are only interested
in determining lower bounds, or the possibility of observing an entailment; for instance, when considering
positive subsumption. Recall once again that, according to our semantics, a contextual GCI $(C\sqsubseteq D)^\kappa$
will hold in any world $\omega$ such that $\omega\not\models\kappa$. Thus, if the probability of this world is
positive ($P_\Bmc(\omega)>0$), we can immediately guarantee that $P_\Kmc((C\sqsubseteq D)^\kappa)>0$.
Thus, positive subsumption can be decided without any ontological reasoning for any context that is not
almost certain. In all other cases, the problem can still be reduced to inconsistency.

\begin{theorem}
\label{thm:possub}
Let $\Kmc=(\Tmc,\Bmc)$ be a consistent KB.
The concept $C$ is positively subsumed by $D$ in context $\kappa$ w.r.t.\ \Kmc  iff 
$\Kmc'=(\Tmc\cup\{C(a)^\kappa,\neg D(a)^\kappa\},\Bmc)$ is inconsistent or $P_\Bmc(\kappa)<1$.
\end{theorem}
\begin{proof}
Assuming that \Kmc is consistent, if $P_\Bmc(\kappa)<1$, the result is an immediate consequence of
Theorem~\ref{thm:sum}. Otherwise, $P_\Bmc(\kappa)=1$, and since \Kmc is consistent, inconsistency of
$\Kmc'$ can only arise from the two new assertions, which means that $\kappa\models\phi_{\Kmc'}^\bot$ and
in particular $\sum_{\omega \models \phi_{\kbDL^\prime}^\bot} P_\Bmc(\omega)=1$.
\end{proof}
This theorem requires the context $\kappa$ to be specified to decide positive subsumption.
If $\kappa$ is not known before hand, it is also possible to leverage the inconsistency decision process, which
is the most expensive part of this method; recall that otherwise we only need $\phi_\Bmc$, which can be
easily computed from \Bmc. Let 
$\Kmc_\emptyset:=(\Tmc\cup\{C(a)^\emptyset,\neg D(a)^\emptyset\},\Bmc)$. That is, $\Kmc_\emptyset$ extends \Kmc 
with assertions negating the subsumption relation, which should hold in all contexts. From 
Theorem~\ref{thm:tab} we conclude that $\phi_{\Kmc_\emptyset}^\bot$ encodes all the contexts in which 
$C\sqsubseteq D$ must hold. Notice that the computation of $\phi_{\Kmc_\emptyset}^\bot$ does not depend
on the context $\kappa$ but can be used to decide positive subsumption for any context.

\begin{corollary}
The concept $C$ is positively subsumed by $D$ in context $\kappa$ w.r.t.\ $\Kmc=(\Tmc,\Bmc)$ iff 
$\kappa$ entails $\phi_{\Kmc_\emptyset}^\bot$ or $P_\Bmc(\kappa)<1$.
\end{corollary}

\medskip

Considering the probabilities of contextual subsumption relations may lead to unexpected or unintuitive results arising from
the contextual semantics. Indeed, it always holds (see Theorem~\ref{thm:possub}) that
$P_\Kmc((C\sqsubseteq D)^\kappa)\ge 1-P_\Bmc(\kappa)$. In other words, the probability of a subsumption
in a very unlikely context will always be very high, regardless of the KB and concepts used. In essence, this is
caused because the semantics of contextual subsumption impose a variation of a material implication---if $\kappa$ holds,
then the subsumption holds---which is vacuously true if the premise (if $\kappa$ holds) is false.
In some cases,
it may be more meaningful to consider a \emph{conditional} probability under the assumption that the context
$\kappa$ holds, as defined next.
\begin{definition}[conditional subsumption]
\label{def:sub:cprobs}
Let $\probModelDL = (\mathcal{J}, P_{\mathcal{J}})$ be a probabilistic model of the KB $\kbDL$, 
$\kappa, \lambda$ two contexts with $P_\Bmc(\lambda)>0$, and
$C,D$ two concepts. The \emph{conditional probability of $(C\sqsubseteq D)^\kappa$ given $\lambda$} 
w.r.t.\ \Pmc is
$$P_{\probModelDL}(\balcStmnt{(C \sqsubseteq D)}{\kappa}\mid \lambda)= 
	\frac{\sum_{\vModelDL \in \mathcal{J}, v^\Vmc\models\lambda,\vModelDL \models \balcStmnt{(C \sqsubseteq D)}{\kappa}} P_{\mathcal{J}}(\vModelDL)}{P_\Bmc(\lambda)}.$$
The \emph{conditional probability of $\balcStmnt{(C \sqsubseteq D)}{\kappa}$ given $\lambda$} w.r.t.\  $\kbDL$ is
$$P_{\kbDL}(\balcStmnt{(C \sqsubseteq D)}{\kappa}\mid\lambda) = 
	\inf_{\mathcal{P \models K}} P_{\mathcal{P}}(\balcStmnt{(C \sqsubseteq D)}{\kappa}\mid\lambda).$$
\end{definition}
This definition follows the same principles of conditioning in probability theory, but extended to the open world
interpretation provided by our model-based semantics. Notice that, in addition to the scaling factor 
$P_\Bmc(\lambda)$ in the denominator, the nominator is also differentiated from Definition~\ref{def:sub:probs}
by considering only the worlds that satisfy the context $\lambda$ already.

\modify{Notice that the conditional probabilities in Definition \ref{def:sub:cprobs} are pretty similar to those in 
Definition~\ref{def:sub:probs}, except that the sum restricts to only the worlds that satisfy the context
$\lambda$.} Thus, the numerator can be obtained from the contextual probability in the context 
$\kappa\land\lambda$ excluding the worlds that violate $\lambda$. More formally, we have that 
\begin{align*}
\qquad 
\sum_{\mathclap{\vModelDL \in \mathcal{J}, v^\Vmc\models\lambda,\vModelDL \models \balcStmnt{(C \sqsubseteq D)}{\kappa}}} P_{\mathcal{J}}(\vModelDL) 
	= {} & \sum_{v^\Vmc\models\lambda,v^\Vmc\not\models\kappa} P_\Jmc(\Vmc) + 
		  \sum_{v^\Vmc\models\lambda\land\kappa,\Vmc\models C\sqsubseteq D} P_\Jmc(\Vmc) \\
	= {} & P_\Bmc(\lambda)-P_\Bmc(\lambda\land\kappa) + 
		P_\Pmc((C\sqsubseteq D)^{\lambda\land\kappa}) - 1 + P_\Bmc(\lambda\land\kappa) \\
	= {} & P_\Pmc((C\sqsubseteq D)^{\lambda\land\kappa}) + P_\Bmc(\lambda) - 1.
\end{align*}
Thus we get the following result.
\begin{theorem}
\label{thm:cond:sub}
$P_\Kmc((C\sqsubseteq D)^\kappa\mid\lambda)
	=\frac{P_\Pmc((C\sqsubseteq D)^{\lambda\land\kappa}) + P_\Bmc(\lambda) - 1}{P_\Bmc(\lambda)}$.
\end{theorem}
In particular, this means that also conditional probabilities can be computed through contextual probabilities,
with a small overhead of computing the probability (on the BN \Bmc) of the conditioning context $\lambda$.

As in the contextual case, if one is only interested in knowing that the subsumption is possible (that is, that
it has a positive probability), then one can exploit the complex context describing the inconsistent contexts which,
as mentioned before, can be precompiled to obtain the contexts in which a subsumption relation must be 
satisfied. However, in this case, entailment between contexts is not sufficient; one must still compute the 
probability of the contextual subsumption.
\begin{corollary}
Let $\Kmc=(\Tmc,\Bmc)$ be a consistent KB, $C,D$ concepts, and $\kappa,\lambda$ contexts s.t.\
$P_\Bmc(\lambda)>0$. $P_\Kmc((C\sqsubseteq D)^\kappa\mid \lambda)>0$ iff
$P_\Pmc((C\sqsubseteq D)^{\lambda\land\kappa}) > 1 - P_\Bmc(\lambda)$.
\end{corollary}
In other words, $P((C\sqsubseteq D)^\kappa\mid \lambda)>0$ iff $C$ is $p$-subsumed by $D$ in 
$\kappa\land\lambda$ with $p=1-P_\Bmc(\lambda)$.

\modify{
\begin{example}
Continuing Example \ref{exa:prob:sub}, suppose that we know that the context $\neg X$ holds; that is, the house 
was built before 1986. Following Theorem \ref{thm:cond:sub} we can compute
\begin{align*}
P_\Kmc((\text{Water}\sqsubseteq\text{Drinkable})^\emptyset \mid \neg X) = {} &
	\frac{P_\Kmc((\text{Water}\sqsubseteq\text{Drinkable})^{\neg X})+P_\Bmc(\neg X)-1}{P_\Bmc(\neg X)}
\end{align*}
The denominator can be read directly from the probability distribution in Figure \ref{fig:BN}; that is, 
$P_\Bmc(\neg X)=0.3$. 
From Theorem \ref{thm:sum}, the numerator becomes the sum of the probabilities of all worlds that become inconsistent
with the addition of the assertions $\{\text{Water}(a)^{\neg X},\neg\text{Drinkable}(a)^{\neg X}\}$. These are the 
same appearing in the last four rows of Table \ref{tab:worlds}. Hence, 
\begin{align*}
P_\Kmc((\text{Water}\sqsubseteq\text{Drinkable})^\emptyset \mid \neg X) = {} &
	\frac{0.216}{0.3} = 0.72.
\end{align*}
As expected, knowing that the house is pre-1986 increases the probability of having lead tubing, which means a decrease
in the probability of drinkable water.
\end{example}
}

\medskip

Analogously to Definitions~\ref{def:con:sub}, \ref{def:sub:probs}, and~\ref{def:sub:cprobs}, it is possible to
define the notions of consistency of a concept $C$ to hold in only some contexts, and based on it, the
(potentially conditional) probability of such a contextual consistency problem. As mentioned already, it is
well known that for every \ALC interpretation \Imc it holds that $C^\Imc=\emptyset$ iff $C^\Imc\subseteq\bot^\Imc$.
Hence, all these problems can be solved through a direct reduction to their related subsumption problem.

We now turn our attention to the problem of instance checking. In this problem, the ABox also plays a role.
Hence, we consider once again ontologies \Omc that can have in addition to GCIs, concept and role assertions.

\section{Instance Checking} 
\label{sec:IC}

We consider a probabilistic extension to the classical instance checking problem. In $\BALC$ we call this problem 
probabilistic instance checking and we define both a decision problem and probability calculation for 
it next, following the same pattern as in the previous section.

\begin{definition}[Instance]
Given a KB \Kmc and a context $\kappa$, the individual name $a$ is an \emph{instance} of the
concept $C$ in $\kappa$ w.r.t.\ \Kmc, written $\kbDL \models \balcStmnt{C(a)}{\kappa}$, iff for all 
probabilistic models $\probModelDL = (\mathcal{J}, P_\mathcal{J})$ of $\kbDL$ and for all $\Vmc\in\Jmc$
it holds that $\Vmc\models C(a)^\kappa$. 
\end{definition}
That is, $a$ is an instance of $C$ in $\kappa$ if every interpretation in each model \Pmc satisfies the assertion $C(a)^\kappa$\!.

Note that as before, instance checking in \ALC is a special case of this definition, that can be obtained 
by considering a BN with only one variable that is true with probability $1$. Notice that, contrary to the case
of satisfiability studied at the end of the previous section, it is not possible to reduce instance checking to subsumption
since an instance may be---and in fact usually is---caused by ABox assertions, which are ignored for subsumption tests.
\modify{The latter observation is not of great consequence though, since both instance checking and subsumption can be reduced to consistency checking.}
\begin{theorem}
Given an individual name $a\in N_I$, a concept $C$, a context $\kappa$, and a KB $\kbDL = (\Omc, \bnDL)$, 
$\Kmc\models C(a)^\kappa$ iff the KB $\Kmc'=(\Omc\cup\{(\neg C(a))^\kappa\},\Bmc)$ is inconsistent.
\end{theorem}
In particular, this means that instance checking is at most as hard as deciding consistency of a KB. As mentioned already,
it is also at least as hard as instance checking in the classical \ALC. Hence we get the following result.
\begin{corollary}
\modify{Instance checking w.r.t.\ \BALC KBs is \ExpTime-complete.}
\end{corollary}
Let us now consider the probabilistic entailments related to instance checking.
\begin{definition}[instance probability]
The \emph{probability of an instance} in a probabilistic model $\probModelDL = (\mathcal{J}, P_{\mathcal{J}})$ of 
the KB $\kbDL$ is
$$
P_{\probModelDL} (\balcStmnt{C(x)}{\kappa}) = 
	\sum_{\vModelDL \in \mathcal{J}, \vModelDL \models \balcStmnt{C(x)}{\kappa}} P_{\mathcal{J}} (\vModelDL).
$$
The \emph{instance probability} w.r.t.\ a KB $\kbDL$ is
$$
P_{\kbDL} (\balcStmnt{C(x)}{\kappa}) = 
	\inf_{\probModelDL \models \kbDL} P_{{\probModelDL}}(\balcStmnt{C(x)}{\kappa}).
$$
The conditional probability of an instance in a particular probabilistic model 
$\probModelDL = (\mathcal{J}, P_\mathcal{J})$, \modify{where $P_\Bmc(\lambda)>0$}, is
$$
P_\probModelDL (\balcStmnt{C(x)}{\kappa} \mid \lambda) = 
	\frac{\sum_{\vModelDL \in \mathcal{J}, \valFuncDL \models \lambda, \vModelDL \models \balcStmnt{C(x)}{\kappa}} P_\mathcal{J}(\vModelDL)}{P_\Bmc(\lambda)},
$$
The probability of the conditional instance in $\kbDL$ is:
$$P_\kbDL (\balcStmnt{C(x)}{\kappa} \mid \lambda) = \inf_{\probModelDL \models \kbDL} P_\probModelDL (\balcStmnt{C(x)}{\kappa} \mid \lambda)$$
\end{definition}
The probability of all instance checks for an inconsistent KB is always $1$ to keep our definitions consistent with 
probability theory, just as we did for the probability of subsumption over an inconsistent KB in the previous section.

As we did for subsumption, we can exploit the reasoning techniques that were developed for deciding inconsistency of a \BALC 
KB to find out also the contextual and conditional probabilities of an instance. Moreover, the method can be 
further optimised in the cases where we are only interested in probabilistic bounds. In particular, we can
adapt Theorem~\ref{thm:cond:sub} to this case, yielding the following result.
\begin{theorem} 
\label{balc:condSubReduc}
$P_\kbDL (\balcStmnt{C(x)}{\kappa} \mid \lambda) = 
	\frac{P_\kbDL (\balcStmnt{C(x) }{{\kappa \land \lambda}}) + P_\Bmc(\lambda) - 1}{P(\lambda)}.$
\end{theorem}
\modify{
\begin{example}
Suppose that we extend the ontology from Example \ref{exa:ont} with the assertions $\text{Pipe}(p)^{\neg X}$,
$\text{Lead}(m)^Z$, and
$\text{contains}(p,m)^{\neg Y}$. To compute $P_\Kmc(\text{LeadPipe}(p)^\emptyset)$ we must find out what
$P_\Bmc(\neg X,\neg Y,Z)$ is. Following the chain rule from the BN, this is computed as
$P(\neg X)P(\neg Y\mid \neg X)P(Z\mid \neg X,\neg Y)=0.3\cdot 0.3\cdot 0.6=0.054$. That is, even though we have
(probabilistic) information about a pipe $p$, which may contain a material which may be lead, the probability that $p$
is a lead pipe is very low.
\end{example}
}
With Theorem \ref{balc:condSubReduc}, we conclude the technical analysis of contextual and probabilistic reasoning tasks over 
\BALC KBs. As it can be readily seen, the general approach is based on being able to detect the contexts 
in which the conditions imposed by the ontology are inconsistent. The probabilistic computations rely on
cross-checking those contexts with the distribution described by the BN.

\section{Conclusions}
\label{sec:conclusions}

We have presented a new probabilistic extension of the DL \ALC, which is based on the ideas of Bayesian ontology
languages. In this setting, knowledge bases contain certain knowledge, which is dependent on the uncertain context where
it holds. Our work extends
the results on \BEL~\cite{CePe-IJCAR14,CePe-JELIA14} to a propositionally closed ontology language. The main
notions follow the basic ideas of general Bayesian ontology languages as presented in \cite{CePe-JAR16}; however, by 
focusing on a specific logic---rather than on a general notion of ontology language---we are able to produce a tableaux-based 
decision algorithm for KB consistency, in contrast
to the generic black-box algorithms which exist in the literature. Our algorithm extends the classical tableau
algorithm for \ALC with techniques originally developed for axiom pinpointing. The main differences in our approach are the
use of multi-valued variables in the definition of the contexts, and the possibility of having complex contexts
(instead of only unique variables) labelling individual axioms.
In general, we have shown that adding 
context-based uncertainty to an ontology does not increase the complexity of reasoning in this logic: all
(probabilistic) reasoning problems can still be solved in exponential time. Of course, this tight complexity bound is not
extremely surprising given the \ExpTime-hardness of the underlying classical DL.

Theorems~\ref{thm:formula}--\ref{thm:bn} yield an effective decision method \modify{for} \BALC KB consistency, through
the computation and handling of two complex contexts, which describe the logical and probabilistic properties
of the worlds and the restricted ontologies they define.
Notice that $\phi_\Bmc$ can be computed in linear time on the size of \Bmc, and satisfiability
of the context $\phi_\Kmc^\bot\land\phi_\Bmc$ can be checked in non-deterministic polynomial time in the size
of this context, following techniques akin to those from propositional satisfiability \cite{BHSVW09}. 
However, the tableau algorithm for computing $\phi_\Kmc^\bot$ is not optimal w.r.t.\ worst-case
complexity. In fact, this algorithm requires double exponential time in the worst case, although the context it computes is only
of exponential size. The benefit of this method, as in the classical case, is that it provides a better behaviour
in the average case; indeed, depending on the input ontology, it may not need to verify all the possible contexts, but simply
construct a model. As a simple example, if the TBox is empty, the process stops in exponential time. 
To further improve the efficiency of our approach, one can think of adapting the methods
from~\cite{ZBRCL18} to construct a compact representation of the context---akin to a binary decision diagram
(BDD)~\cite{Lee-59,BrRB-90} for multi-valued variables---allowing for efficient weighted model counting. A
task for future work is to exploit these data structures for practical development of our methods.

Recall that the most expensive part of our approach is the computation of the context $\phi_\Kmc^\bot$. By
slightly modifying the KB, we have shown that one computation of this context suffices to solve different
problems of interest; in particular, contextual and conditional entailments---being subsumption, satisfiability, or
instance checking---can be solved using $\phi_{\Kmc'}^\bot$, for an adequately constructed $\Kmc'$,
regardless of the contexts under consideration. 

An important next step will be to implement the methods described here, and compare the efficiency of
our system to other probabilistic DL reasoners based on similar semantics. In particular, we would like to
compare against the tools from~\cite{ZBRCL18}. Even though this latter system is also based on an extension
of the tableaux algorithm for DLs, and use multiple-world semantics closely related to ours, a direct 
comparison would be unfair. Indeed,~\cite{ZBRCL18} makes use of stronger independence assumptions than
ours. However, a well-designed experiment can shed light on the advantages and disadvantages of each
method.

Another interesting problem for future work is to extend the query language beyond instance queries. To maintain
some efficiency, this may require some additional restrictions on the language or the probabilistic structure. A more
detailed study of this issue is needed.


\begin{thebibliography}{}

\bibitem[\protect\citeauthoryear{Artale, Calvanese, Kontchakov, and
  Zakharyaschev}{Artale et~al\mbox{.}}{2009}]{ACKZ09}
{\sc Artale, A.}, {\sc Calvanese, D.}, {\sc Kontchakov, R.}, {\sc and} {\sc
  Zakharyaschev, M.} 2009.
\newblock The dl-lite family and relations.
\newblock {\em J. Artif. Intell. Res.\/}~{\em 36}, 1--69.

\bibitem[\protect\citeauthoryear{Baader, Brandt, and Lutz}{Baader
  et~al\mbox{.}}{2005}]{BaBL05}
{\sc Baader, F.}, {\sc Brandt, S.}, {\sc and} {\sc Lutz, C.} 2005.
\newblock Pushing the {EL} envelope.
\newblock In {\em Proceedings of the Nineteenth International Joint Conference
  on Artificial Intelligence ({IJCAI} 2005)}, {L.~P. Kaelbling} {and}
  {A.~Saffiotti}, Eds. Professional Book Center, 364--369.

\bibitem[\protect\citeauthoryear{Baader, Calvanese, McGuinness, Nardi, and
  Patel-Schneider}{Baader et~al\mbox{.}}{2007}]{BCea-07}
{\sc Baader, F.}, {\sc Calvanese, D.}, {\sc McGuinness, D.}, {\sc Nardi, D.},
  {\sc and} {\sc Patel-Schneider, P.}, Eds. 2007.
\newblock {\em The Description Logic Handbook: Theory, Implementation, and
  Applications\/}, Second ed.
\newblock Cambridge University Press.

\bibitem[\protect\citeauthoryear{Baader and Hollunder}{Baader and
  Hollunder}{1995}]{BaHo-95}
{\sc Baader, F.} {\sc and} {\sc Hollunder, B.} 1995.
\newblock Embedding defaults into terminological knowledge representation
  formalisms.
\newblock {\em J. Autom. Reasoning\/}~{\em 14,\/}~1, 149--180.

\bibitem[\protect\citeauthoryear{Baader, Horrocks, Lutz, and Sattler}{Baader
  et~al\mbox{.}}{2017}]{BHLS17}
{\sc Baader, F.}, {\sc Horrocks, I.}, {\sc Lutz, C.}, {\sc and} {\sc Sattler,
  U.} 2017.
\newblock {\em An Introduction to Description Logic}.
\newblock Cambridge University Press.

\bibitem[\protect\citeauthoryear{Baader and Pe{\~n}aloza}{Baader and
  Pe{\~n}aloza}{2007}]{BaPe07}
{\sc Baader, F.} {\sc and} {\sc Pe{\~n}aloza, R.} 2007.
\newblock Axiom pinpointing in general tableaux.
\newblock In {\em Proceedings of the 16th International Conference on Analytic
  Tableaux and Related Methods ({TABLEAUX~2007})}, {N.~Olivetti}, Ed. Lecture
  Notes in Artificial Intelligence, vol. 4548. Springer-Verlag,
  Aix-en-Provence, France, 11--27.

\bibitem[\protect\citeauthoryear{Baader and Pe{\~n}aloza}{Baader and
  Pe{\~n}aloza}{2010}]{BaPe-JLC09}
{\sc Baader, F.} {\sc and} {\sc Pe{\~n}aloza, R.} 2010.
\newblock Axiom pinpointing in general tableaux.
\newblock {\em Journal of Logic and Computation\/}~{\em 20,\/}~1 (February),
  5--34.
\newblock Special Issue: Tableaux and Analytic Proof Methods.

\bibitem[\protect\citeauthoryear{Biere, Heule, van Maaren, and Walsh}{Biere
  et~al\mbox{.}}{2009}]{BHSVW09}
{\sc Biere, A.}, {\sc Heule, M.}, {\sc van Maaren, H.}, {\sc and} {\sc Walsh,
  T.} 2009.
\newblock {\em Handbook of Satisfiability: Volume 185 Frontiers in Artificial
  Intelligence and Applications}.
\newblock IOS Press, Amsterdam, The Netherlands, The Netherlands.

\bibitem[\protect\citeauthoryear{Botha}{Botha}{2018}]{Both-18}
{\sc Botha, L.} 2018.
\newblock The {Bayesian} description logic {ALC}.
\newblock M.S.\ thesis, University of Cape Town, South Africa.

\bibitem[\protect\citeauthoryear{Botha, Meyer, and Pe{\~{n}}aloza}{Botha
  et~al\mbox{.}}{2018}]{BoMP18}
{\sc Botha, L.}, {\sc Meyer, T.}, {\sc and} {\sc Pe{\~{n}}aloza, R.} 2018.
\newblock The {B}ayesian description logic {BALC}.
\newblock In {\em Proceedings of the 31st International Workshop on Description
  Logics (DL 2018)}, {M.~Ortiz} {and} {T.~Schneider}, Eds. {CEUR} Workshop
  Proceedings, vol. 2211. CEUR-WS.org.

\bibitem[\protect\citeauthoryear{Botha, Meyer, and Pe{\~{n}}aloza}{Botha
  et~al\mbox{.}}{2019}]{BoMP19}
{\sc Botha, L.}, {\sc Meyer, T.}, {\sc and} {\sc Pe{\~{n}}aloza, R.} 2019.
\newblock A bayesian extension of the description logic \emph{ALC}.
\newblock In {\em Proceedings of the 16th European Conference on Logics in
  Artificial Intelligence ({JELIA} 2019)}, {F.~Calimeri}, {N.~Leone}, {and}
  {M.~Manna}, Eds. Lecture Notes in Computer Science, vol. 11468. Springer,
  339--354.

\bibitem[\protect\citeauthoryear{Brace, Rudell, and Bryant}{Brace
  et~al\mbox{.}}{1990}]{BrRB-90}
{\sc Brace, K.~S.}, {\sc Rudell, R.~L.}, {\sc and} {\sc Bryant, R.~E.} 1990.
\newblock Efficient implementation of a bdd package.
\newblock In {\em Proceedings of the 27th ACM/IEEE Design Automation
  Conference}. DAC '90. ACM, New York, NY, USA, 40--45.

\bibitem[\protect\citeauthoryear{Ceylan}{Ceylan}{2018}]{Ceyl-18}
{\sc Ceylan, {\.I}.~{\.I}.} 2018.
\newblock Query answering in probabilistic data and knowledge bases.
\newblock Ph.D. thesis, Dresden University of Technology, Germany.

\bibitem[\protect\citeauthoryear{Ceylan and Lukasiewicz}{Ceylan and
  Lukasiewicz}{2018}]{CeLu-18}
{\sc Ceylan, {\.I}.~{\.I}.} {\sc and} {\sc Lukasiewicz, T.} 2018.
\newblock A tutorial on query answering and reasoning over probabilistic
  knowledge bases.
\newblock In {\em 14th International Summer School on Reasoning Web},
  {C.~d'Amato} {and} {M.~Theobald}, Eds. Lecture Notes in Computer Science,
  vol. 11078. Springer, 35--77.

\bibitem[\protect\citeauthoryear{Ceylan and Pe{\~n}aloza}{Ceylan and
  Pe{\~n}aloza}{2014}]{CePe-IJCAR14}
{\sc Ceylan, {\.I}.~{\.I}.} {\sc and} {\sc Pe{\~n}aloza, R.} 2014.
\newblock The {B}ayesian description logic {BEL}.
\newblock In {\em Proceedings of the 7th International Joint Conference on
  Automated Reasoning (IJCAR 2014)}, {S.~Demri}, {D.~Kapur}, {and}
  {C.~Weidenbach}, Eds. Lecture Notes in Computer Science, vol. 8562. Springer,
  480--494.

\bibitem[\protect\citeauthoryear{Ceylan and Pe{\~{n}}aloza}{Ceylan and
  Pe{\~{n}}aloza}{2014}]{CePe-JELIA14}
{\sc Ceylan, {\.I}.~{\.I}.} {\sc and} {\sc Pe{\~{n}}aloza, R.} 2014.
\newblock Tight complexity bounds for reasoning in the description logic bel.
\newblock In {\em Proceedings of the 14th European Conference on Logics in
  Artificial Intelligence ({JELIA} 2014)}, {E.~Ferm{\'{e}}} {and} {J.~Leite},
  Eds. Lecture Notes in Computer Science, vol. 8761. Springer, 77--91.

\bibitem[\protect\citeauthoryear{Ceylan and Pe{\~n}aloza}{Ceylan and
  Pe{\~n}aloza}{2017}]{CePe-JAR16}
{\sc Ceylan, {\.I}.~{\.I}.} {\sc and} {\sc Pe{\~n}aloza, R.} 2017.
\newblock The {Bayesian} ontology language {BEL}.
\newblock {\em Journal of Automated Reasoning\/}~{\em 58,\/}~1, 67--95.

\bibitem[\protect\citeauthoryear{d'Amato, Fanizzi, and Lukasiewicz}{d'Amato
  et~al\mbox{.}}{2008}]{AmFL08}
{\sc d'Amato, C.}, {\sc Fanizzi, N.}, {\sc and} {\sc Lukasiewicz, T.} 2008.
\newblock Tractable reasoning with bayesian description logics.
\newblock In {\em Proceedings of the Second International Conference on
  Scalable Uncertainty Management}, {S.~Greco} {and} {T.~Lukasiewicz}, Eds.
  Lecture Notes in Computer Science, vol. 5291. Springer, 146--159.

\bibitem[\protect\citeauthoryear{Darwiche}{Darwiche}{2009}]{Darw09}
{\sc Darwiche, A.} 2009.
\newblock {\em Modeling and Reasoning with Bayesian Networks}.
\newblock Cambridge University Press.

\bibitem[\protect\citeauthoryear{Donini and Massacci}{Donini and
  Massacci}{2000}]{DoMa-00}
{\sc Donini, F.~M.} {\sc and} {\sc Massacci, F.} 2000.
\newblock Exptime tableaux for $\mathcal{ALC}$.
\newblock {\em Artificial Intelligence\/}~{\em 124,\/}~1, 87--138.

\bibitem[\protect\citeauthoryear{Gottlob, Lukasiewicz, Martinez, and
  Simari}{Gottlob et~al\mbox{.}}{2013}]{GottlobLMS13}
{\sc Gottlob, G.}, {\sc Lukasiewicz, T.}, {\sc Martinez, M.~V.}, {\sc and} {\sc
  Simari, G.~I.} 2013.
\newblock Query answering under probabilistic uncertainty in datalog+/-
  ontologies.
\newblock {\em Annals of Mathematics and Artificial Intelligence\/}~{\em
  69,\/}~1, 37--72.

\bibitem[\protect\citeauthoryear{Guti{\'{e}}rrez{-}Basulto, Jung, Lutz, and
  Schr{\"{o}}der}{Guti{\'{e}}rrez{-}Basulto et~al\mbox{.}}{2017}]{GJLR-17}
{\sc Guti{\'{e}}rrez{-}Basulto, V.}, {\sc Jung, J.~C.}, {\sc Lutz, C.}, {\sc
  and} {\sc Schr{\"{o}}der, L.} 2017.
\newblock Probabilistic description logics for subjective uncertainty.
\newblock {\em J. Artif. Intell. Res.\/}~{\em 58}, 1--66.

\bibitem[\protect\citeauthoryear{Halpern}{Halpern}{1990}]{Halp-AIJ90}
{\sc Halpern, J.~Y.} 1990.
\newblock An analysis of first-order logics of probability.
\newblock ~{\em 46,\/}~3, 311--350.

\bibitem[\protect\citeauthoryear{Lee}{Lee}{1959}]{Lee-59}
{\sc Lee, C.~Y.} 1959.
\newblock Representation of switching circuits by binary-decision programs.
\newblock {\em The Bell System Technical Journal\/}~{\em 38}, 985--999.

\bibitem[\protect\citeauthoryear{Lee, Meyer, Pan, and Booth}{Lee
  et~al\mbox{.}}{2006}]{LMPB06}
{\sc Lee, K.}, {\sc Meyer, T.~A.}, {\sc Pan, J.~Z.}, {\sc and} {\sc Booth, R.}
  2006.
\newblock Computing maximally satisfiable terminologies for the description
  logic {ALC} with cyclic definitions.
\newblock In {\em Proceedings of the 2006 International Workshop on Description
  Logics (DL2006)}, {B.~Parsia}, {U.~Sattler}, {and} {D.~Toman}, Eds. {CEUR}
  Workshop Proceedings, vol. 189. CEUR-WS.org.

\bibitem[\protect\citeauthoryear{Lukasiewicz and Straccia}{Lukasiewicz and
  Straccia}{2008}]{LuSt-JWS08}
{\sc Lukasiewicz, T.} {\sc and} {\sc Straccia, U.} 2008.
\newblock Managing uncertainty and vagueness in description logics for the
  semantic web.
\newblock {\em Journal of Web Semantics\/}~{\em 6,\/}~4, 291--308.

\bibitem[\protect\citeauthoryear{Meyer, Lee, Booth, and Pan}{Meyer
  et~al\mbox{.}}{2006}]{MLBP06}
{\sc Meyer, T.~A.}, {\sc Lee, K.}, {\sc Booth, R.}, {\sc and} {\sc Pan, J.~Z.}
  2006.
\newblock Finding maximally satisfiable terminologies for the description logic
  {ALC}.
\newblock In {\em Proceedings of the Twenty-First National Conference on
  Artificial Intelligence and the Eighteenth Innovative Applications of
  Artificial Intelligence Conference}. {AAAI} Press, 269--274.

\bibitem[\protect\citeauthoryear{Pearl}{Pearl}{1985}]{Pear-85}
{\sc Pearl, J.} 1985.
\newblock Bayesian networks: A model of self-activated memory for evidential
  reasoning.
\newblock In {\em Proc. of Cognitive Science Society (CSS-7)}. 329--334.

\bibitem[\protect\citeauthoryear{Pe{\~n}aloza}{Pe{\~n}aloza}{2009}]{PenaPhD}
{\sc Pe{\~n}aloza, R.} 2009.
\newblock Axiom-pinpointing in description logics and beyond.
\newblock Ph.D. thesis, Dresden University of Technology, Germany.

\bibitem[\protect\citeauthoryear{Schild}{Schild}{1991}]{Schi-91}
{\sc Schild, K.} 1991.
\newblock A correspondence theory for terminological logics: {P}reliminary
  report.
\newblock In {\em Proceedings\ of the 12th International\ Joint Conference\ on
  Artificial Intelligence (IJCAI~1991)}, {J.~Mylopoulos} {and} {R.~Reiter},
  Eds. Morgan Kaufmann, 466--471.

\bibitem[\protect\citeauthoryear{Schmidt{-}Schau{\ss} and
  Smolka}{Schmidt{-}Schau{\ss} and Smolka}{1991}]{ScSc-91}
{\sc Schmidt{-}Schau{\ss}, M.} {\sc and} {\sc Smolka, G.} 1991.
\newblock Attributive concept descriptions with complements.
\newblock {\em Artificial Intelligence\/}~{\em 48,\/}~1, 1--26.

\bibitem[\protect\citeauthoryear{Zese, Bellodi, Riguzzi, Cota, and Lamma}{Zese
  et~al\mbox{.}}{2018}]{ZBRCL18}
{\sc Zese, R.}, {\sc Bellodi, E.}, {\sc Riguzzi, F.}, {\sc Cota, G.}, {\sc and}
  {\sc Lamma, E.} 2018.
\newblock Tableau reasoning for description logics and its extension to
  probabilities.
\newblock {\em Ann. Math. Artif. Intell.\/}~{\em 82,\/}~1-3, 101--130.

\end{thebibliography}
\end{document}